\documentclass[11pt]{article}
\usepackage[T1]{fontenc}
\usepackage[utf8]{inputenc}
\usepackage{pbox}

\usepackage{tikz}
\usetikzlibrary{shapes}
\usepackage{pgfplots}
\pgfplotsset{width=10cm,compat=1.9}

\usepgfplotslibrary{external}
\usepackage{tikz}
\usetikzlibrary{patterns,decorations.pathreplacing,calc,fadings,calc,shapes.callouts,shapes.arrows,backgrounds,spy,shadows,decorations.markings}
\immediate\write18{mkdir -p tikz}

\usetikzlibrary{intersections}
\usetikzlibrary{positioning}
\usetikzlibrary{decorations.pathreplacing}
\usetikzlibrary{matrix}
\usetikzlibrary{calc}
\usetikzlibrary{arrows,arrows.meta}
\usetikzlibrary{arrows,shapes,trees,backgrounds,shadows}
\usetikzlibrary{decorations.pathreplacing}
\usetikzlibrary{decorations.pathmorphing} 
\usetikzlibrary{backgrounds,snakes}
\usetikzlibrary{decorations.markings}
\usetikzlibrary{positioning}
\usetikzlibrary{plotmarks}
\usetikzlibrary{trees}
\usetikzlibrary{patterns}
\usetikzlibrary{spy}
\usetikzlibrary{calc,matrix}

\usepackage{scalerel}

\usepackage{graphicx} 
\usepackage[margin=1in]{geometry}
\usepackage{amsmath}
\usepackage{caption}
\usepackage{hyperref}
\usepackage{subcaption}
\usepackage{xcolor}
\usepackage{amssymb}
\usepackage[numbers]{natbib}
\usepackage{multirow}
\usepackage{soul}

\usepackage{amsmath, amsfonts, dsfont}
\usepackage{amsthm}
\usepackage{hyperref}
\usepackage{graphicx}
\usepackage{placeins}
\usepackage[ruled]{algorithm2e}

\newtheorem{theorem}{Theorem}
\newtheorem{lemma}{Lemma}
\newtheorem{proposition}{Proposition}

\newtheorem{remark}{Remark}

\newcommand{\argmax}{\mathrm{argmax}}

\SetKwInput{KwInputs}{Inputs}
\SetKwInput{KwOutputs}{Outputs}
\SetKwRepeat{Repeat}{repeat}{until}

\newcommand{\apass}{\vec{\mathbf{0}}}

\usepackage{xcolor}

\begin{document}




\title{Optimal Batched Scheduling of Stochastic Processing Networks Using Atomic Action Decomposition}


\author{%
    Jim Dai\\
    \footnotesize{Operations Research and Information Engineering, Cornell University, jd694@cornell.edu}\and
    Manxi Wu\\
    \footnotesize{Department of Civil and Environmental Engineering, University of California, Berkeley, manxiwu@berkeley.edu}\and
    Zhanhao Zhang\\
    \footnotesize{Operations Research and Information Engineering, Cornell University, zz564@cornell.edu}
}

\date{}
\maketitle

\begin{abstract}
    Stochastic processing networks (SPNs) have broad applications in healthcare, transportation, and communication networks. The control of SPN is to dynamically assign servers in batches under uncertainty to optimize long-run performance. This problem is challenging as the policy dimension grows exponentially with the number of servers, making standard reinforcement learning and policy optimization methods intractable at scale.
We propose an atomic action decomposition framework that addresses this scalability challenge by breaking joint assignments into sequential single-server assignments. This yields policies with constant dimension, independent of the number of servers. We study two classes of atomic policies, the step-dependent and step-independent atomic policies, and prove that both achieve the same optimal long-run average reward as the original joint policies. These results establish that computing the optimal SPN control can be made scalable without loss of optimality using the atomic framework. Our results offer theoretical justification for the strong empirical success of the atomic framework in large-scale applications reported in previous articles. 

\end{abstract}



%


\section{Introduction} \label{sec:intro}
The stochastic processing network (SPN) model (\cite{dai2020processing}) has been widely used to model complex service systems in domains such as healthcare, communication, and transportation. These networks consist of multiple servers, each capable of executing certain service types, and the system manager’s goal is to dynamically assign servers over time to optimize long-term performance under uncertainty. While early work focuses on developing control policies that are optimal or near-optimal under certain limiting regimes (e.g. heavy traffic), their performance may degrade substantially outside the limiting regimes.  Reinforcement learning (RL) offers a promising alternative, but applying standard RL methods to SPNs is challenging due to the high dimensionality of the state and action spaces, especially in large networks.

In this work, we develop new reinforcement learning algorithms for the optimal control of stochastic processing networks (SPNs) under batched decision-making, where control actions are chosen at discrete time intervals over groups of jobs. Batched control naturally arises in applications such as shift scheduling, inventory restocking, and hospital workforce planning, where decisions must account for procurement delays and logistical constraints. Batched control also enables the system to aggregate more information before committing to actions, which can improve assignment quality and reduce inefficiencies \citep{akbarpour2020thickness}. However, optimizing over batched actions leads to a combinatorial explosion in the number of feasible server assignments as the system scales, limiting the scalability of standard reinforcement learning methods and motivating the need for new algorithmic approaches.

We consider a stochastic processing network (SPN) with multiple servers, item classes, and service types. Items of each class arrive exogenously and enter class-specific buffers, where they wait to be processed or are actively in service. Each item class is eligible for specific service types, and each service type can be performed by a compatible subset of servers. Unlike prior models that assume fixed compatibility between servers and services \citet{dai2020processing}, we allow compatibility to be dynamic, evolving based on each server’s service history. This generalization captures practical settings such as ride-hailing markets, where a driver’s ability to serve passengers in a certain region depends on the driver’s current location. We formulate the SPN control problem as a Markov Decision Process (MDP), where the system state includes current item and service counts, and each action specifies a complete assignment of servers to services. The reward at each time step equals the total service reward minus the holding cost of all queued items, and the objective is to find a policy that maximizes the long-run average reward.

A key challenge in this setting is the exponential growth of the action space with the number of servers, which makes direct policy optimization intractable for large systems. To address this, we propose the {\em atomic action decomposition} framework, which breaks the joint assignment into sequential assignments of individual servers. We define each single-server assignment as an {\em atomic action}, and its contribution to the objective as the {\em atomic reward}. At each time step, the system selects multiple atomic actions one by one, updating the state after each atomic assignment to reflect the new service count. Once all servers are assigned, the holding cost is applied, and the system transitions to the next state to account for item arrivals and service completions. By decomposing the decision process into a sequence of {\em atomic steps}, this framework reduces the policy dimension to depend only on the number of feasible assignments for a single server, which does not scale with the server number.

We consider two types of atomic policies: step-dependent and step-independent. A step-dependent atomic policy takes as input both the system state and the atomic step index (i.e., the number of atomic actions already generated within the current time step), allowing it to use a distinct decision rule at each step. In contrast, a step-independent atomic policy takes only the system state as input and applies the same decision rule at every atomic step. Both types maintain a constant policy dimension that does not scale with the number of servers. The step-independent policy avoids state augmentation and therefore yields a substantially smaller overall policy space, particularly beneficial in large-scale systems where the number of atomic steps is large.

We prove that  both types of atomic policies can achieve the same optimal long-run average reward as the original policies, as stated in Theorems \ref{thm:atomic-optimal-aug} and \ref{thm:atomic-policy-equivalence}, respectively. \emph{These results are significant because they show that optimal control in SPNs as an essential class of models for service and operations systems can be computed using the highly scalable atomic decision framework without loss of optimality.} The proofs of Theorems \ref{thm:atomic-optimal-aug} and \ref{thm:atomic-policy-equivalence} leverage the {\em reward decomposability} property of the SPN control problems, where the reward (or cost) incurred by the joint assignment of all servers can be expressed as the sum of rewards (or costs) for the assignment of individual servers. In particular, Theorem \ref{thm:atomic-optimal-aug} is established by showing that the solution to the optimality equation of the atomic MDP (i.e., the MDP induced by step-dependent atomic policies) also satisfies the optimality equation of the original MDP. The proof of Theorem~\ref{thm:atomic-policy-equivalence} constructs a reduced passing-last atomic MDP and shows that it is equivalent to the original atomic MDP in terms of long-run average reward for any deterministic step-independent policy. This construction enables us to show that the optimality equation of the passing-last atomic MDP is equivalent to that of the original MDP, which implies that the step-independent atomic policy achieves the optimal reward of the original policies.

Our theoretical results help explain the empirical success of Atomic-PPO, a deep reinforcement learning algorithm that integrates atomic action decomposition into proximal policy optimization \citep{schulman2017proximal}. In Section~\ref{sec:applications}, we review results from prior work \citep{sun2024inpatient, Huo_Switch_Scheduling_via, dai2025atomicEV, feng2021scalable} showing that Atomic-PPO performs well in domains such as hospital inpatient overflow assignment, input-queued switch scheduling, and ride-hailing dispatching, all of which can be modeled as SPNs with high-dimensional control spaces. The reward decomposition property is satisfied in all  of these three prior work. We leave it to future work to explore the robustness of the atomoic-PPO algorithm for MDP problems when the reward decompostion is only approximately satisfied.

\subsection{Literature review}
The optimal control of stochastic processing networks has been studied extensively for more than 30 years. Early work focuses on developing control policies that are optimal or near-optimal under certain limiting regimes (e.g. heavy traffic) \citep{pedarsani2017robust, pedarsani2014scheduling, pedarsani2014robust, liu2025zero, liu2022large, liu2019spatial, liu2013scheduling, bodas2012low, jaramillo2011optimal, dong2024shortest, dong2021srpt, berg2024asymptotically, choudhury2021job, yang2023learning, weng2020achieving, weng2022algorithm, lu2016optimal, li2016queue, bodas2009scheduling, athanasopoulou2012back, ni2011q, ying2011cluster, eryilmaz2005stable, liu2011throughput, bertsimas2014robust, chen1993dynamic, harrison1996bigstep, meyn1997stability, maglaras2000discrete, ata2005heavy, bauerle2000asymptotic, lu1994efficient, kumar1996fluctuation, dai2005maximum, tassiulas1995adaptive, tassiulas1992jointly, stolyar2004maxweight, miyazawa2017unified}. While these approaches provide important structural insights, their performance may degrade substantially outside the limiting regimes. In recent years, various reinforcement learning (RL) approaches have been adopted to SPN problems, including value-based algorithms \citep{dong2025multiclass, dai2019inpatient, moallemi2008approximate, brown2001switch, liu2022rl, chen2018deep, oshri2017deep, wang2018deep, oroojlooyjadid2022deep, vanvuchelen2025use}, policy-based algorithms \citep{peshkin2002reinforcement, yu2018drom}, and actor-critic methods \citep{dai2022queueing, vanvuchelen2020use, van2023using, kaynov2024deep, liu2022multi, ding2022multi}. The direct application of standard RL techniques to large-scale networks faces severe scalability challenges due to the high-dimensional state and action spaces. Our atomic action decomposition framework contributes to this line of literature by reducing the high-dimension SPN control problem to a sequence of low-dimensional decisions over individual server assignments, enabling scalable policy learning without sacrificing optimality.

Our idea of atomic action decomposition relates to cooperative multi-agent reinforcement learning (MARL), where multiple agents act based on a shared global state to maximize a common reward, with both the reward and state transitions determined by their joint actions \citep{kuba2021trust, bertsekas2019multiagent, gu2021multi, kuba2022heterogeneous, wang2023order, kuba2021settling, wen2022multi, fu2022revisiting, ye2022towards, wen2024reinforcing}. As in our setting, the policy space in cooperative MARL grows exponentially with the number of agents, posing major challenges for policy training. A common strategy to mitigate this complexity is to factor the joint policy into a sequence of lower-dimensional policies, where each agent’s action is generated conditioned on the actions of preceding agents \citep{wen2022multi, fu2022revisiting, ye2022towards, wen2024reinforcing}. This class of policies, often referred to as {\em auto-regressive policies}, retains optimality when each agent’s action is conditioned on the pre-decision state and all preceding actions \citep{ye2022towards}.

Our atomic policy framework significantly improves scalability over the auto-regressive approaches by exploiting the structural properties of SPN control problems. Unlike auto-regressive policies, which require each agent (or server) to condition its action on the entire history of previously selected actions, our atomic policies avoid conditioning on action histories altogether. Step-dependent atomic policies use only the current state and step index, while step-independent policies rely solely on the current state. As a result, the input space of our atomic policies remains fixed across steps and independent of the number of servers, whereas the input space of auto-regressive policies grows exponentially with the number of agents.

%

Our reduction in atomic policy input dimension is achieved without sacrificing optimality, thanks to two key features unique to SPNs. First, state transitions across atomic steps can be explicitly defined through deterministic updates to service counts, eliminating the need to track prior actions. Second, the reward decomposability property ensures that the total reward can be written as a sum of per-server rewards, allowing the system to evaluate each atomic decision independently.
In contrast, general MARL settings often lack these properties: state transitions are typically entangled across agents, and policy updates require conditioning on full action histories.
This necessitates auto-regressive policies to condition on action histories, leading to high and variable input dimensions. As a result, auto-regressive policies in MARL often rely on complex architectures such as multi-head attention layers \citep{vaswani2023attentionneed}, whereas our atomic PPO framework only requires shallow feed-forward neural networks, making it far more scalable in large systems.


Beyond the connections with MARL, our work also relates to a distinct line of research that enhances policy training efficiency by exploiting structural properties of specific classes of SPNs \citep{alvo2023neural, madeka2022deep, che2024differentiable,jali2024efficient}. Alvo et al. and Madeka et al. leverage the differentiability of action and transition functions in inventory control problems, Che et al. develop a policy gradient estimation method for discrete-event simulations of multi-class queueing networks, while Jali et al. adopt a parameterized threshold policy to tackle with the routing problems in parallel server systems. In contrast, our atomic action decomposition leverages the state-transition and reward decomposability properties of SPNs.


The remainder of this paper is organized as follows. Section \ref{sec:model} introduces the SPN control problem and formulates it as an MDP. Section \ref{sec:atomic-decomp} presents our atomic action decomposition framework, defines two forms of atomic policies, and proves that both can achieve the optimal long-run average reward. In Section \ref{sec:atomic-ppo}, we introduce our Atomic-PPO algorithm. Section \ref{sec:applications} illustrates how the SPN model can be applied to hospital inpatient overflow assignment, switch scheduling, and ride-hailing, and how their control policy optimization benefits from our atomic action decomposition framework. Finally, Section \ref{sec:conclusion} concludes the paper and discusses future research directions. We include supplementary statements and proofs in the appendix.

\section{Control of stochastic processing networks} \label{sec:model}
We consider a non-preemptive stochastic processing network with $K$ servers, $J$ service types, and $I$ item classes. We use $\mathcal{J}$ to denote the set of all service types and $\mathcal{I}$ to denote the set of all item classes.
Items of each class $i \in \mathcal{I}$ arrive exogenously into the system. The newly arrived item will be queued in the unique buffer associated with its class until it is processed by a server. Each class of items can be processed by services of certain types. In particular, we define the material requirement matrix $B \in \{0, 1\}^{I \times J}$, where $B_{ij} = 1$ indicates that an item of class $i \in \mathcal{I}$ can be processed by a service of type $j \in \mathcal{J}$. We require that 
\begin{subequations} \label{eq:B-matrix-constraints}
\begin{align}
    \sum_{i \in \mathcal{I}} B_{ij} \leq& 1, &\forall j \in \mathcal{J}, \label{eq:B-matrix-constraints-1}\\
    \sum_{j \in \mathcal{J}} B_{ij} \geq& 1, &\forall i \in \mathcal{I}, \label{eq:B-matrix-constraints-2}
\end{align}
\end{subequations}
where \eqref{eq:B-matrix-constraints-1} ensures that each type of service can process at most one class of item, and \eqref{eq:B-matrix-constraints-2} ensures that each item class can be processed by at least one type of service.

After a service completes, the associated item is removed from its current buffer and then either routed to the buffer of a different class or exits the system. We use $F \in \{0, 1\}^{I \times J}$ to denote the item routing matrix, where $F_{ij} = 1$ indicates that an item will be added to buffer $i \in \mathcal{I}$ after the completion of a type $j \in \mathcal{J}$ service. If, for a given service type $j \in \mathcal{J}$, $F_{ij} = 0$ for all item classes $i \in \mathcal{I}$, then the item processed by the type $j$ service will exit the system upon service completion.

A server is required to execute each service. We say that a service is {\em open} if it has been initiated but not completed. After a service completes, the associated server becomes {\em idle} and is available to receive a new service of a certain type. We define the compatibility matrix $C \in \{0, 1\}^{J \times J}$, where $C_{j'j} = 1$ indicates that a server completing a type $j'$ service is eligible to begin a type $j$ service. We require that every service type can be executed following some prior service:
\begin{align*}
    \sum_{j' \in \mathcal{J}} C_{j'j} \geq& 1, &\forall j \in \mathcal{J}.
\end{align*}

The stochastic processing network evolves on discrete time steps indexed by $t = 0, 1, \dots$, where each time step is a fixed interval during which the system is updated. At each time $t$, the system keeps track of the current item counts $z^t := (z^t_{i})_{i \in \mathcal{I}}$ and the current service counts $n^t := (n^t_{j,\tau})_{j \in \mathcal{J}, \tau \in \mathbb{N}_+ \cup \{\infty\}}$, where $z^t_{i}$ represents the number of items in buffer $i \in \mathcal{I}$ (including both items waiting to be served and items currently being served) and $n^t_{j, \tau}$ represents the number of type $j \in \mathcal{J}$ services with an {\em age} of $\tau \in \mathbb{N}_+ \cup \{\infty\}$. The {\em age} of a service represents the number of time steps it has remained open. 

Each open service of type $j \in \mathcal{J}$ with an age of $\tau \in \mathbb{N}_+$ has a probability of $\mu_{j, \tau} \in [0, 1]$ to complete in the current time step. We use $n^t_{j',\infty}$ to denote the number of idle servers at time $t$ that were previously executing a type $j' \in \mathcal{J}$ service.
We note that the total number of servers equals the sum of busy and idle servers, that is,
\begin{align*}
    \sum_{j \in \mathcal{J}} \sum_{\tau \in \mathbb{N}_+ \cup \{\infty\}} n^t_{j,\tau} = K.
\end{align*}


At each time $t$, based on the item count vector $z^t := (z^t_{i})_{i \in \mathcal{I}}$ and the service count vector $n^t := (n^t_{j,\tau})_{j \in \mathcal{J}, \tau \in \mathbb{N}_+ \cup \{\infty\}}$, the {\em system manager} selects a batch of services to initiate, which we refer to as a {\em schedule}. We denote the schedule at time $t$ as $a^t := (a^t_{j'j})_{j,j' \in \mathcal{J}}$, where $a^t_{j'j}$ represents the number of servers that previously executed service type $j'$ and are assigned to execute service type $j$. A feasible schedule $a^t$ needs to satisfy the following constraints:
\begin{enumerate}
    \item The server assignment is compatible \eqref{eq:setup-feasibility-compatible} and the total number of assigned servers does not exceed the available server number \eqref{eq:setup-feasibility-server}: 
    \begin{align} 
        &a^t_{j'j} > 0 \text{ only if } C_{j'j} = 1, &\forall j,j' \in \mathcal{J}, \label{eq:setup-feasibility-compatible} \\
        &\sum_{j \in \mathcal{J}} a^t_{j'j} \leq n^t_{j', \infty}, &\forall j' \in \mathcal{J}. \label{eq:setup-feasibility-server}
    \end{align}
    \item For each item class $i \in \mathcal{I}$, the total number of class $i$ items processed by all services—including the newly initiated and open ones—cannot exceed the number of items in buffer $i$:
    \begin{align} \label{eq:setup-feasibility-item}
        &\sum_{j,j' \in \mathcal{J}} a^t_{j'j} B_{ij} + \sum_{j \in \mathcal{J}} \sum_{\tau \in \mathbb{N}_+} n^t_{j,\tau} B_{ij} \leq z^t_{i}, &\forall i \in \mathcal{I}.
    \end{align}
    In \eqref{eq:setup-feasibility-item}, the first term on the left-hand side represents the total number of items processed by newly initiated services, and the second term represents the total number of items processed by open services initiated in previous time steps.
\end{enumerate}

We denote the post-decision time as $t_+$ and update the item counts and service counts to $z^{t_+}$ and $n^{t_+}$ to reflect the initiation of new services:
\begin{align}
    n^{t_+}_{j,0} =& \sum_{j' \in \mathcal{J}} a^t_{j'j}, &\forall j \in \mathcal{J}, \label{eq:setup-service-transition}\\
    n^{t_+}_{j,\infty} =& n^t_{j, \infty} - \sum_{j' \in \mathcal{J}} a^t_{jj'}, &\forall j \in \mathcal{J}. \label{eq:setup-service-transition-inf}
\end{align}
The remaining entries of $(z^{t_+}, n^{t_+})$ equal to that of $(z^t, n^t)$.
Equation \eqref{eq:setup-service-transition} sets the counts of age $0$ services to be equal to the newly initiated ones. Equation \eqref{eq:setup-service-transition-inf} decrements the counts of idle servers to reflect their assignment to new services. 

We denote $\Lambda_i$ as the distribution of exogenous arrivals for class $i$ items in each time step and $x_i^t$ as the number of class $i$ arrivals realized during time $[t, t+1)$. Let $y_{j,\tau}^t$ denote the number of completed type $j$ services with age $\tau$ during $[t, t+1)$. Since each of the $n^{t_+}_{j,\tau}$ services completes independently with probability $\mu_{j,\tau}$, $y_{j,\tau}^t$ follows the distribution $\text{Binomial}(n^{t_+}_{j,\tau}, \mu_{j,\tau})$. Then, the system updates its item counts and service counts to reflect the new item arrivals and service completions:
\begin{align}
    z^{t+1}_{i} =& z^{t_+}_{i} + x^{t}_i - \sum_{j \in \mathcal{J}} \sum_{\tau \in \mathbb{N}_+} y^{t}_{j,\tau} B_{ij} \notag\\
    &+ \sum_{j \in \mathcal{J}} \sum_{\tau \in \mathbb{N}_+} y^{t}_{j,\tau} F_{ij}, \notag\\
    &\quad\forall i \in \mathcal{I}, \label{eq:setup-item-transition-dynamic} \\
    n^{t+1}_{j,\tau} =& n^{t_+}_{j,\tau-1} - y^{t}_{j,\tau-1},\notag\\
    &\quad \forall j \in \mathcal{J},\ \forall \tau > 0, \label{eq:setup-service-transition-dynamic} \\
    n^{t+1}_{j,\infty} =& n^{t_+}_{j,\infty} + \sum_{\tau \in \mathbb{N}_+} y^{t}_{j,\tau},\notag\\
    &\quad \forall j \in \mathcal{J}. \label{eq:setup-service-transition-inf-dynamic}
\end{align}
Equation \eqref{eq:setup-item-transition-dynamic} updates the item counts in each buffer based on three components: (i) the number of class $i$ items newly arrived into the system $x^{t+1}_i$; (ii) the number of class $i$ items exiting the buffer after service completion $\sum_{j \in \mathcal{J}} \sum_{\tau \in \mathbb{N}_+} y^{t}_{j,\tau-1} B_{ij}$; and (iii) and the number of items routed to buffer $i$ after service completion $\sum_{j \in \mathcal{J}} \sum_{\tau \in \mathbb{N}_+} y^{t}_{j,\tau-1} F_{ij}$. Equations \eqref{eq:setup-service-transition-dynamic}-\eqref{eq:setup-service-transition-inf-dynamic} updates the service counts by accounting for the service completions and age increase of open services during time $[t,t+1)$.

\subsection{Markov decision process}
We model the system manager's control of the stochastic processing network as a discrete-time Markov Decision Process (MDP) with an infinite time horizon.
The state space is denoted by $\mathcal{S}$, where each state $s^t := (z^t, n^t) \in \mathcal{S}$ records the item counts and service counts for each time $t$. For ease of exposition, we assume that the state space is finite. This finiteness is ensured by imposing upper bounds on both the total number of items in the system and the maximum service age. Specifically, we assume that newly arrived items will be rejected if the total item count exceeds a large but finite threshold. Moreover, for each service type $j \in \mathcal{J}$, we assume the service completion probability $\mu_{j,\tau}$ equals $1$ for all ages $\tau$ greater than or equal to some large but finite value.
The action space is denoted by $\mathcal{A}$, where each action corresponds to a schedule in the SPN model.
For each state $s := (z, n) \in \mathcal{S}$, we denote the set of feasible actions $\mathcal{A}_s \subseteq \mathcal{A}$ as the set of all actions $a$ that satisfy constraints \eqref{eq:setup-feasibility-compatible}–\eqref{eq:setup-feasibility-item}.

We denote the transition matrix $P \in [0,1]^{|\mathcal{S}| \times |\mathcal{A}| \times |\mathcal{S}|}$, where $P(s^{t+1} \mid s^t, a^t)$ is the probability of transitioning from state $s^t$ at time $t$ to state $s^{t+1}$ at time $t+1$, given action $a^t \in \mathcal{A}_{s^t}$. We use $P^{\text{sys}} \in [0,1]^{|\mathcal{S}| \times |\mathcal{S}|}$ to denote the transition matrix corresponding to the stochastic update. Specifically, $P^{\text{sys}}(s^{t+1} \mid s^{t_+})$ represents the probability of transitioning from the {\em post-decision state} $s^{t_+} := (n^{t_+}, z^{t_+})$ at time $t$ to the state $s^{t+1}$ at time $t+1$ due to exogenous dynamics as in \eqref{eq:setup-item-transition-dynamic}–\eqref{eq:setup-service-transition-inf-dynamic}.

We define the reward function $r: \mathcal{S} \times \mathcal{A} \rightarrow \mathbb{R}$, where $r(s^t, a^t)$ represents the reward generated by choosing the action $a^t \in \mathcal{A}_{s^t}$ given the state $s^t \in \mathcal{S}$. The reward has two components: a one-time reward (or cost) $r_j \in \mathbb{R}$ incurred when initiating a service of type $j \in \mathcal{J}$, and a per time-step holding cost $r_H(s^{t_+})$ for items in the buffer, which depends on the post-decision state $s^{t_+}$. Here, $r_H: \mathbb{N}_+^{\vert \mathcal{S} \vert} \rightarrow \mathbb{R}_{\leq 0}$ denotes the holding cost function. The reward $r(s^t, a^t)$ is given by:
\begin{align} \label{eq:setup-schedule-reward}
    r(s^t, a^t) := \sum_{j \in \mathcal{J}} \sum_{j' \in \mathcal{J}} a^t_{j'j} r_j + r_{H}( s^{t_+}_{} ).
\end{align}

We denote $\pi: \mathcal{S} \rightarrow \Delta(\mathcal{A})$ as the policy of the system manager. We assume that the MDP is an aperiodic unichain, meaning that under any policy, the induced Markov chain has a single recurrent class (possibly with transient states) and is aperiodic (see Page 348 of \cite{PutermanMDP}). Under this assumption, the Markov chain admits a unique stationary distribution for each policy $\pi$, allowing us to define its long-run average reward as: \[
    R(\pi) = \lim_{T \rightarrow \infty} \frac{1}{T} \mathbb{E}_{\pi}\left[\sum_{t = 1}^{T} r(s^t, a^t) \right],
\] The goal of the system manager is to find a policy that maximizes the long-run average reward over an infinite time horizon. We denote the optimal value as $R^* := \sup_{\pi} R(\pi)$ and the corresponding optimal policy as $\pi^*$.


\section{Atomic action decomposition} \label{sec:atomic-decomp}
One core challenge for computing the optimal policy $\pi^*$ is the exponential growth of the action space $\mathcal{A}$ as the number of servers $K$ increases. For each idle server, a decision must be made based on both the type $j' \in \mathcal{J}$ of the last completed service and the type $j \in \mathcal{J}$ of the next service to initiate.  There are up to \( J \times J = J^2 \) such (previous type, next type) compatible combinations for each server. Since a schedule can assign new services to any subset of up to $K$ servers, the total number of actions grows on the order of \( O((J^2)^K) = O(J^{2K}) \). Optimizing a policy over this large action space is computationally prohibitive.

To reduce this complexity, we propose an atomic decomposition of actions. We define an \emph{atomic action} $\hat{a}$ as a binary matrix $\hat{a} \in \{0, 1\}^{J \times J}$ that contains at most one non-zero entry, i.e., $\sum_{j,j' \in \mathcal{J}} \hat{a}_{j'j} \leq 1$. 
If \( \hat{a}_{j'j} = 1 \), then the atomic action $\hat{a}$ indicates assigning a server that previously completed a type \( j' \) service to initiate a type \( j \) service. If all entries are zero, then the atomic action $\hat{a}$ is referred as the \emph{passing} atomic action, denoted by \( \apass \), indicating that no service is initiated. The set of all atomic actions, denoted by $\hat{\mathcal{A}}$, has size  \( O(J^2) \), which does not scale with the number of servers. Given a system state $s$,  an atomic action is feasible if it satisfies constraints \eqref{eq:setup-feasibility-compatible}-\eqref{eq:setup-feasibility-item}. We denote the set of feasible atomic actions given state $s$ as \( \hat{\mathcal{A}}_s \subseteq \hat{\mathcal{A}} \).

At each time step, the system manager sequentially selects \( K \) atomic actions, referred to as \( K \) \emph{atomic steps}. Consider time \( t \), the state before assigning the first atomic action, denoted as $s_1^t$, is the same as the system state of time $t$, i.e. $s_1^t = s^t$. Based on this state, the system manager selects a feasible atomic action \( \hat{a}^t_1 \in \hat{\mathcal{A}}_{s^t_1} \), and the system transitions to a new state \( s^t_2 := (z^t_2, n^t_2) \) that reflects the service initiated by the selected atomic action. This procedure is repeated for each atomic step \( k = 1, \dots, K \). Let \( s^t_k := (z^t_k, n^t_k) \in \mathcal{S} \) be the system state before the \( k \)-th atomic step, and \( \hat{a}^t_k \in \hat{\mathcal{A}}_{s^t_k} \) be the k-th atomic action. The state transition from \( s^t_k \) to \( s^t_{k+1} \) is given by:
\begin{align}
    (z^t_{k+1})_i =& (z^t_k)_i, &\forall i \in \mathcal{I}, \label{eq:atomic-item-transition}\\
    (n^t_{k+1})_{j,0} =& (n^t_k)_{j,0} + \sum_{j' \in \mathcal{J}} (\hat{a}^t_{k})_{j'j}, & \forall j \in \mathcal{J}, \label{eq:atomic-service-transition-0}\\
    (n^t_{k+1})_{j,\infty} =& (n^t_k)_{j,\infty} - \sum_{j' \in \mathcal{J}} (\hat{a}^t_k)_{jj'}, &\forall j \in \mathcal{J}. \label{eq:atomic-service-transition-inf}
\end{align}
After assigning all \( K \) atomic actions, the system reaches post-decision state \( s^t_{K+1} \). The transition from \( s_{K+1}^{t} \) to the state of time $t+1$ incorporates the item arrival and service completion, and follows the same update given by \eqref{eq:setup-item-transition-dynamic}–\eqref{eq:setup-service-transition-inf-dynamic}. The sequential state transitions induced by these \( K \) atomic steps are illustrated in Figure~\ref{fig:atomic-steps}.

\begin{figure}[ht]
    \centering
    \includegraphics[width=0.7\linewidth]{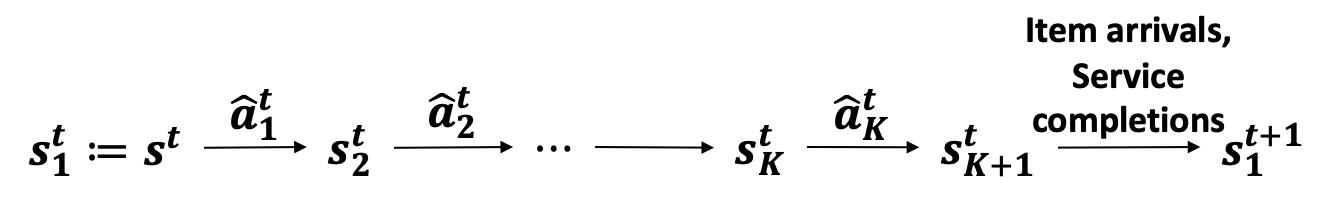}
    \caption{State transitions induced by atomic actions in each time $t$.}
    \label{fig:atomic-steps}
\end{figure}

This sequence of atomic steps defines a new MDP. In particular, the transition matrix for each atomic step \( k \in [K] \) is a binary matrix $\hat{P} := \left( \hat{P}(s^t_{k+1} \mid s^t_k, \hat{a}^t_k) \right)_{s^t_k, \hat{a}^t_k, s^t_{k+1}}$,
where each entry \( \hat{P}(s^t_{k+1} \mid s^t_k, \hat{a}^t_k) \in \{0, 1\}\) indicates whether or not the system transitions to state \( s^t_{k+1} \) from state \( s^t_k \) with the atomic action \( \hat{a}^t_k \) under the deterministic update rules \eqref{eq:atomic-item-transition}–\eqref{eq:atomic-service-transition-inf}.
After completing all \( K \) atomic steps, the system transitions from \( s^t_{K+1} \) to the state \( s^{t+1}_1 \) at time $t+1$ to account for new item arrivals and service completions, where we recall that the transition probability from $s^t_{K+1}$ to $s^{t+1}_1$ is given by $P^{\text{sys}}(s^{t+1}_1 \mid s^t_{K+1})$.

We define $\hat{r}(\hat{a}^t_k)$ as the reward obtained by the atomic action $\hat{a}^t_k \in \hat{\mathcal{A}}$ at atomic step $k \in [K]$ of time $t$: \[
    \hat{r}(\hat{a}^t_k) = \sum_{j \in \mathcal{J}} r_j \cdot \left(\sum_{j' \in \mathcal{J}} (\hat{a}^t_k)_{j'j} \right).
\]

We note that any sequence of atomic actions $\{\hat{a}^t_k \in \hat{\mathcal{A}}\}_{k \in [K]}$ induces an {\em equivalent} action $a^t \in \mathcal{A}$, where each entry satisfies $a^t_{j'j} = \sum_{k \in [K]} (\hat{a}^t_k)_{j'j}$ for all $j', j \in \mathcal{J}$. Conversely, for any action $a^t \in \mathcal{A}$, one can construct an equivalent sequence of atomic actions that induces the same schedule. 
The one-step reward given action $a^t$ and state $s^t$ is equal to the sum of the rewards obtained from each atomic action, plus the holding cost:
\begin{align} \label{eq:reward-equiv}
    r(s^t, a^t) = \sum_{k \in [K]} \hat{r}(\hat{a}^t_k) + r_H(s^t_{K+1}).
\end{align}

We consider two types of atomic policies:
\begin{itemize}
    \item[-] The {\em step-dependent} atomic policy: $\hat{\pi} := \{\hat{\pi}_k: \mathcal{S} \rightarrow \Delta(\hat{\mathcal{A}})\}_{k \in [K]}$ uses a distinct policy $\hat{\pi}_k$ to generate the $k$-th atomic at each atomic step. 
    The long-run average reward under a step-dependent atomic policy $\hat{\pi}$ is given by:
    \begin{align*}
        R(\hat{\pi}) =& \lim_{T \to \infty} \frac{1}{T} \mathbb{E}_{\hat{\pi}}\left[\sum_{t=1}^T \right.\\
        &\left.\left(\sum_{k=1}^{K} \hat{r}( \hat{a}_k^t) + r_H(s^t_{K+1}) \right) \right].
    \end{align*}
    We denote the optimal step-dependent atomic policy as $\hat{\pi}^*$ and the optimal long-run average reward as $R(\hat{\pi}^*)$.
    \item[-] The {\em step-independent atomic policy}: $\tilde{\pi}: \mathcal{S} \rightarrow \Delta(\hat{\mathcal{A}})$ uses the same policy function across all atomic steps. 
    The long-run average reward under a step-independent atomic policy $\tilde{\pi}$ is given by: 
    \begin{align*}
        R(\tilde{\pi}) =& \lim_{T \to \infty} \frac{1}{T} \mathbb{E}_{\tilde{\pi}}\left[\sum_{t=1}^T \right.\\
        &\left.\left(\sum_{k=1}^{K} \hat{r}( \hat{a}_k^t) + r_H(s^t_{K+1}) \right) \right].
    \end{align*}
    We denote the optimal step-independent atomic policy as $\tilde{\pi}^*$ and the optimal long-run average reward as $R(\tilde{\pi}^*)$.
\end{itemize}

We now present the main results of our work.
Theorem \ref{thm:atomic-optimal-aug} establishes that the optimal long-run average reward $R(\hat{\pi}^*)$ of step-dependent atomic policies equals the optimal long-run average reward $R(\pi^*)$ of the original policies. 

\begin{theorem} \label{thm:atomic-optimal-aug}
$R(\hat{\pi}^*) = R(\pi^*)$.
\end{theorem}

Building on Theorem \ref{thm:atomic-optimal-aug}, Theorem \ref{thm:atomic-policy-equivalence} further shows that there exists an optimal atomic policy that $R(\pi^*)$ can also be attained using a step-independent atomic action policy.

\begin{theorem} \label{thm:atomic-policy-equivalence}
$R(\tilde{\pi}^*) = R(\pi^*)$.
\end{theorem}

In Section \ref{sec:atomic-decomp-optimal}, we prove Theorem \ref{thm:atomic-optimal-aug}, and in Section \ref{sec:atomic-decomp-efficient}, we prove Theorem \ref{thm:atomic-policy-equivalence}.

\subsection{Proof of Theorem \ref{thm:atomic-optimal-aug}} \label{sec:atomic-decomp-optimal}
The proof of Theorem \ref{thm:atomic-optimal-aug} has two steps. First, Lemma \ref{lemma:opt-eq-joint-aug} presents the optimality equation for the MDP under step-dependent atomic policies. Second, Lemma \ref{lemma:opt-eq-aug} shows that the relative value function and the long-run average reward associated with the optimal step-dependent atomic policy satisfy the optimality equation of the original MDP. This implies that the optimal long-run average reward achieved by the step-dependent atomic policies is the same as that of the original policies.


We note that the MDP with atomic action decomposition has period $K$, corresponding to the number of atomic steps per time step. Applying Theorem 8.4.5 (p. 361 of \cite{PutermanMDP}), we derive the optimality equation, as stated in Lemma \ref{lemma:opt-eq-joint-aug} below.
\begin{lemma}[\cite{PutermanMDP}] \label{lemma:opt-eq-joint-aug}
    There exists a constant $\hat{g}^*$ and a set of functions $\hat{h}_k^*: \mathcal{S} \rightarrow \mathbb{R}$ for $k \in [K]$ that satisfy the following optimality equations for the MDP under step-dependent atomic policies:  
    \begin{align}
        \hat{h}_k^*(s_k) =& \left\{ 
        \begin{array}{ll}
            \max_{\hat{a}_k \in \hat{\mathcal{A}}_{s_k}}\left\{\hat{r}(\hat{a}_k) - \frac{1}{K} \hat{g}^* + \right.\\
            \quad \left.\sum_{s_{k+1} \in \mathcal{S}} \hat{P}(s_{k+1} \vert s_k, \hat{a}_k)\hat{h}_{k+1}^*(s_{k+1}) \right\}, \\
            \quad\quad \text{if } k < K, \\
            \max_{\hat{a}_K \in \hat{\mathcal{A}}_{s_K}}\left\{\hat{r}(\hat{a}_K) - \frac{1}{K} \hat{g}^* + \right.\\ 
            \quad \sum_{s_{K+1} \in \mathcal{S}} \hat{P}(s_{K+1} \vert s_K, \hat{a}_K) r_H(s_{K+1}) +\\
            \quad \left.\sum_{s'_1, s_{K+1} \in \mathcal{S}} P^{\text{sys}}(s'_1 \vert s_{K+1}) \hat{h}_{1}^*(s'_1) \right\}, \\
            \quad\quad \text{if } k = K,
        \end{array} \right. \notag\\
        & \hfill \forall s_k \in \mathcal{S}. \label{eq:opt-eq-atomic-aug}
    \end{align}
    The following step-dependent atomic policy $\hat{\pi}^* := \{\hat{\pi}^*_k: \mathcal{S} \rightarrow \Delta(\hat{\mathcal{A}})\}_{k \in [K]}$ is optimal, and achieves the optimal average reward $R(\hat{\pi}^*)= \hat{g}^*$:
    \begin{align} \label{eq:pi-star-atomic-aug}
        &\hat{\pi}_k^{*}(\hat{a}_k \vert s_k) := \begin{cases} 
            1, & \text{if } \hat{a}_k = \argmax_{\hat{a} \in \hat{\mathcal{A}}_{s}}\\
            &\left\{\hat{r}(\hat{a}) - \frac{1}{K} \hat{g}^* \right. + \sum_{s_{k+1} \in \mathcal{S}} \\
            &\ \left. \hat{P}(s_{k+1} \vert s_k, \hat{a})\hat{h}_{k+1}^{*}(s_{k+1}) \right\},\\
            0, & \text{otherwise,}\\
        \end{cases}\notag\\
        & \hfill \forall s_k \in \mathcal{S}, ~ \forall \hat{a}_k \in \hat{\mathcal{A}}_{s_k},\ \forall k < K, \notag\\ 
        &\hat{\pi}_K^{*}(\hat{a}_K \vert s_K) := \begin{cases} 
            1, & \text{if } \hat{a}_K = \argmax_{\hat{a} \in \hat{\mathcal{A}}_{s_K}}\\
            &\left\{\hat{r}(\hat{a}) - \frac{1}{K} \hat{g}^* + \sum_{s_{K+1} \in \mathcal{S}} \right.\\
            &\hat{P}(s_{K+1} \vert s_K, \hat{a}) r_H(s_{K+1}) +\\
            &\sum_{s'_1, s_{K+1} \in \mathcal{S}} \\
            &\left. P^{\text{sys}}(s'_1 \vert s_{K+1}) \hat{h}_{1}^*(s'_1) \right\},\\
            0, & \text{otherwise,}\\
        \end{cases} \notag\\
        & \hfill \forall s_K \in \mathcal{S}, \ \forall \hat{a}_K \in \hat{\mathcal{A}}_{s_K}.
    \end{align}
\end{lemma}

In Lemma \ref{lemma:opt-eq-aug}, we show that the solution $(\hat{g}^*, \hat{h}^*)$ to \eqref{eq:opt-eq-atomic-aug} also satisfies the optimality equation for the original MDP. This implies that the optimal long-run average reward achieved by the optimal step-dependent atomic policy is also the optimal long-run average reward achieved by the optimal original policy. Hence, we can conclude Theorem \ref{thm:atomic-optimal-aug}.
\begin{lemma} \label{lemma:opt-eq-aug}
Define $h^*: \mathcal{S} \rightarrow \mathbb{R}$ such that $h^*(s) := \hat{h}_1^*(s)$ for all $s \in \mathcal{S}$. Then, $(h^*, \hat{g}^*)$ also satisfies the optimality equation for the original MDP:
    \begin{align} 
        &h^*(s) = \notag\\
        &\max_{a \in \mathcal{A}_{s}} \left\{r(s, a) -\hat{g}^* + \sum_{s' \in \mathcal{S}} P(s' \vert s, a) h^{*}(s') \right\},\notag\\
        &\quad \forall s \in \mathcal{S}. \label{eq:orig-atomic-policy-construct-h-body-new-aug}
    \end{align}
    Furthermore, the following deterministic original policy $\pi^*$ achieves the optimal average reward $R(\pi^*)= \hat{g}^*$:
    \begin{align}
        \pi^{*}(a \vert s) =& \begin{cases}
            1, \quad &\text{if } a = \argmax_{a \in \mathcal{A}_{s}}\\
            &\left\{r(s, a) -\hat{g}^* + \sum_{s' \in \mathcal{S}} \right.\\
            &\quad \left. P(s' \vert s, a) h^*(s') \right\},\\
            0, \quad &\text{otherwise},
        \end{cases} \notag\\
        &\hfill \forall s \in \mathcal{S}. \label{eq:orig-atomic-policy-construct-pi-aug}
    \end{align}
\end{lemma}

\begin{proof}{Proof of Lemma \ref{lemma:opt-eq-aug}}
    To prove \eqref{eq:orig-atomic-policy-construct-h-body-new-aug}, we first argue that $h^*(s) \leq \max_{a \in \mathcal{A}_{s}} \left\{r(s, a) -\hat{g}^* + \sum_{s' \in \mathcal{S}} P(s' \vert s, a) h^{*}(s') \right\}$. Consider the optimal deterministic step-dependent atomic policy $\hat{\pi}^*$ that satisfies \eqref{eq:pi-star-atomic-aug}.
    For any state $s \in \mathcal{S}$, let $s_1, \dots, s_{K+1}$ with $s_1 = s$ be the sequence of states induced by $\hat{\pi}^{*}$, and let $\hat{a}^{\dagger}_1, \dots, \hat{a}^{\dagger}_{K}$ be the sequence of atomic actions. 
    Let $a^{\dagger}$ be the action {\em equivalent} to $(\hat{a}^{\dagger}_1, \dots, \hat{a}^{\dagger}_K)$ such that $a^{\dagger}_{j'j} = \sum_{k \in [K]} (\hat{a}^{\dagger}_k)_{j'j}, \ \forall j,j' \in \mathcal{J}$. Then, we have:
    \begin{align}
        &h^*(s) 
        = \hat{h}_1^*(s) \notag\\
        =& \sum_{k = 1}^K \hat{r}(\hat{a}^{\dagger}_k) + r_H(s_{K+1}) - \hat{g}^* \notag\\
        &\quad+ \sum_{s' \in \mathcal{S}} P^{\text{sys}}(s' \vert s_{K+1}) \hat{h}_1^*(s') \tag{a}\\
        =& \sum_{k = 1}^K \hat{r}(\hat{a}^{\dagger}_k) + r_H(s_{K+1}) - \hat{g}^* \notag\\
        &\quad+ \sum_{s' \in \mathcal{S}} P^{\text{sys}}(s' \vert s_{K+1}) \hat{P}(s_{K+1} \vert s_K, \hat{a}^{\dagger}_K) \cdot \notag\\
        &\quad\hat{P}(s_K \vert s_{K-1}, \hat{a}^{\dagger}_{K-1}) \dots \hat{P}(s_2 \vert s_1, \hat{a}^{\dagger}_1) \hat{h}_1^*(s') \tag{b}\\
        =& r(s, a^{\dagger}) - \hat{g}^* + \sum_{s' \in \mathcal{S}} P(s' \vert s, a^{\dagger}) h^*(s') \tag{c}\\
        \leq& \max_{a \in \mathcal{A}_{s}} \left\{r(s, a) -\hat{g}^* + \sum_{s' \in \mathcal{S}} P(s' \vert s, a) h^{*}(s') \right\}, \notag
    \end{align}
    where (a) is obtained by applying \eqref{eq:opt-eq-atomic-aug} recursively for $K$ steps; (b) is due to the fact that the state transitions across the $K$ atomic steps are deterministic, i.e. $\hat{P}(s_{K+1} \vert s_{K}, \hat{a}^{\dagger}_{K}) = \dots = \hat{P}(s_2 \vert s_1, \hat{a}^{\dagger}_1) = 1$; (c) is because $a^{\dagger}$ is equivalent to $(\hat{a}^{\dagger}_1, \dots, \hat{a}^{\dagger}_K)$ so the reward satisfies \eqref{eq:reward-equiv}. 
    
    Next, we show that $h^*(s) \geq \max_{a \in \mathcal{A}_{s}} \left\{r(s, a) -\hat{g}^* + \sum_{s' \in \mathcal{S}} P(s' \vert s, a) h^{*}(s') \right\}$. Let $a^*$ be the action generated by $\pi^*$ given state $s$. Let $(\hat{a}_1^*, \dots, \hat{a}_K^*)$ be any sequence of atomic actions that is equivalent to $a^*$ and let $(s_1^*, \dots, s_{K+1}^*)$ with $s_1^* = s$ be the sequence of induced states. Then, we have:
    \begin{align}
        &\max_{a \in \mathcal{A}_{s}} \left\{r(s, a) -\hat{g}^* + \sum_{s' \in \mathcal{S}} P(s' \vert s, a) h^{*}(s') \right\} \notag\\
        =& r(s, a^*) - \hat{g}^* + \sum_{s' \in \mathcal{S}} P(s' \vert s, a^*) h^*(s') \tag{a}\\
        =& \sum_{k = 1}^K \hat{r}(\hat{a}_k^*) + r_H(s_{K+1}^*) - \hat{g}^* \notag \\
        &\quad+ \sum_{s' \in \mathcal{S}} P^{\text{sys}}(s' \vert s_{K+1}^*) \hat{h}_1^*(s'), \tag{b}
    \end{align}
    where (a) is because $a^*$ is generated by the policy $\pi^*$ that satisfies \eqref{eq:orig-atomic-policy-construct-pi-aug}; (b) is because $a^*$ is equivalent to $(\hat{a}_1^*, \dots, \hat{a}_K^*)$.

    Since the state transition from $s_K^*$ to $s_{K+1}^*$ is deterministic, i.e. $\hat{P}(s_{K+1}^* \vert s_K^*, \hat{a}_K^*) = 1$ and $\hat{P}(s_{K+1} \vert s_K^*, \hat{a}_K^*) = 0$ for all $s_{K+1} \in \mathcal{S}$ such that $s_{K+1} \neq s_{K+1}^*$, we obtain:
    \begin{align}
        & \sum_{k = 1}^K \hat{r}(\hat{a}_k^*) + r_H(s_{K+1}^*) - \hat{g}^* \notag \\
        &\quad+ \sum_{s' \in \mathcal{S}} P^{\text{sys}}(s' \vert s_{K+1}^*) \hat{h}_1^*(s') \notag\\
        =& \sum_{k = 1}^{K-1} \hat{r}(\hat{a}_k^*) - \frac{K-1}{K} \hat{g}^* \notag\\ 
        &\quad + \hat{r}(\hat{a}_K^*) - \frac{1}{K} \hat{g}^* \notag \\
        &\quad + \sum_{s_{K+1} \in \mathcal{S}} \hat{P}(s_{K+1} \vert s_K^*, \hat{a}_K^*) r_H(s_{K+1}^*) \notag \\
        &\quad+ \sum_{s' \in \mathcal{S}} P^{\text{sys}}(s' \vert s_{K+1}^*) \hat{h}_1^*(s').
    \end{align}
    Then, we have:
    \begin{align*}
        & \sum_{k = 1}^{K-1} \hat{r}(\hat{a}_k^*) - \frac{K-1}{K} \hat{g}^* \notag\\ 
        &\quad + \hat{r}(\hat{a}_K^*) - \frac{1}{K} \hat{g}^* \notag \\
        &\quad + \sum_{s_{K+1} \in \mathcal{S}} \hat{P}(s_{K+1} \vert s_K^*, \hat{a}_K^*) r_H(s_{K+1}^*) \notag \\
        &\quad+ \sum_{s' \in \mathcal{S}} P^{\text{sys}}(s' \vert s_{K+1}^*) \hat{h}_1^*(s')\\
        \leq& \sum_{k=1}^{K-1} \hat{r}(\hat{a}_k^*) - \frac{K-1}{K} \hat{g}^* + \notag\\
        &\max_{\hat{a} \in \hat{\mathcal{A}}_{s_{K}^*}}\left\{\hat{r}(\hat{a}) - \frac{1}{K} \hat{g}^* \right. \notag\\
        &+ \sum_{s_{K+1} \in \mathcal{S}} \hat{P}(s_{K+1} \vert s_{K}^*, \hat{a}) r_H(s_{K+1}) \notag\\
        &+ \left. \sum_{s', s_{K+1} \in \mathcal{S}} \hat{P}(s_{K+1} \vert s_K, \hat{a}) P^{\text{sys}}(s' \vert s_{K+1}) \hat{h}_1^*(s') \right\} \notag \\
        =& \sum_{k=1}^{K-1} \hat{r}(\hat{a}_k^*) - \frac{K-1}{K} \hat{g}^* + \hat{h}_K^*(s_K^*), \tag{a}
    \end{align*}
    where (a) is obtained by \eqref{eq:opt-eq-atomic-aug}.
    
    Repeating the same argument for $K-2$ steps, we obtain:
    \begin{align*}
        & \sum_{k=1}^{K-1} \hat{r}(\hat{a}_k^*) - \frac{K-1}{K} \hat{g}^* + \hat{h}_K^*(s_K^*)\\
        \leq& \max_{\hat{a} \in \hat{\mathcal{A}}_{s_1^*}} \left\{\hat{r}(\hat{a}) - \frac{1}{K} \hat{g}^* \right. \notag\\
        &\quad + \left. \sum_{s' \in \mathcal{S}} \hat{P}(s' \vert s_1^*, \hat{a}) \hat{h}^*_2(s') \right\} \\
        =& \hat{h}_1^*(s_1^*) \\
        =& h^*(s_1^*) = h^*(s).
    \end{align*}
    Therefore, we can obtain that
    \begin{align*}
        &h^*(s) \\
        \geq& \max_{a \in \mathcal{A}_{s}} \left\{r(s, a) -\hat{g}^* + \sum_{s' \in \mathcal{S}} P(s' \vert s, a) h^{*}(s') \right\}.
    \end{align*}
    
    As a result, we have
    \begin{align*} 
        &h^*(s) \\
        =& \max_{a \in \mathcal{A}_{s}} \left\{r(s, a) -\hat{g}^* + \sum_{s' \in \mathcal{S}} P(s' \vert s, a) h^{*}(s') \right\},\\
        &\quad \forall s \in \mathcal{S}.
    \end{align*}

    Consequently, following Theorem 8.4.4 (p. 361 of \cite{PutermanMDP}), the policy $\pi^*$ given by \eqref{eq:orig-atomic-policy-construct-pi-aug} achieves the optimal long-run average reward $R(\pi^*) = \hat{g}^*$.
    \hfill $\square$
\end{proof}

\subsection{Proof of Theorem \ref{thm:atomic-policy-equivalence}} \label{sec:atomic-decomp-efficient}
We prove Theorem \ref{thm:atomic-policy-equivalence} by constructing a deterministic step-independent atomic policy that attains the optimal long-run average reward $R(\pi^*)$. Consider any deterministic policy $\tilde{\pi}: \mathcal{S} \rightarrow \hat{\mathcal{A}}$, for any atomic step $k \in [K]$, if the generated atomic action is $\apass$ (i.e. passing), then the atomic reward is $0$ and $s_{k+1} = s_k$. Since the policy is deterministic and the state has not changed, the atomic action at the atomic step $k+1$ is still $\apass$ and the state remains unchanged. This will continue for all subsequent atomic steps until $k = K$. Therefore, given any deterministic step-independent atomic policy, once a passing atomic action is generated, the state will remain unchanged for the rest of the atomic steps. This observation allows us to restrict our attention to a reduced model, the {\em passing-last} atomic MDP, where the system terminates the current time step and transitions to the next time step whenever the passing atomic action is generated. Specifically, under the passing-last atomic MDP, the system evolves according to \eqref{eq:atomic-item-transition}–\eqref{eq:atomic-service-transition-inf} when a non-passing atomic action is generated, and according to \eqref{eq:setup-item-transition-dynamic}–\eqref{eq:setup-service-transition-inf-dynamic} when $\apass$ is generated.
We formally show the equivalence between the passing-last atomic MDP and the original atomic MDP in Lemma \ref{lemma:atomic-mdp-equivalence}. This equivalence allows us to carry out the rest of the proof within the passing-last atomic MDP without loss of generality.

In Lemma~\ref{lemma:atomic-h-exist}, we construct the relative value function $\tilde{h}^*$ for the passing-last atomic MDP from the solution $(\hat{g}^*, h^*)$ to the optimality equation \eqref{eq:orig-atomic-policy-construct-h-body-new-aug} of the original MDP given in Lemma \ref{lemma:opt-eq-aug}, and show that $\tilde{h}^*$ is well-defined.

\begin{lemma} \label{lemma:atomic-h-exist}
     Let $(\hat{g}^*, h^*)$ be the solution to \eqref{eq:orig-atomic-policy-construct-h-body-new-aug}. There exists a function $\tilde{h}^{*}: S \to \mathbb{R}$ that satisfies: 
    \begin{align} 
        &\tilde{h}^{*}(s) =\notag\\
        &\max_{\hat{a} \in \hat{\mathcal{A}}_{s}}\left\{ \left[r_H(s) - \hat{g}^* + \sum_{s' \in \mathcal{S}} P^{\text{sys}}(s' \vert s) h^*(s') \right]\mathds{1}_{\hat{a} = \apass} \right. \notag\\
        &\quad+ \left.\left[\hat{r}(\hat{a}) + \sum_{s^{'} \in \mathcal{S}} \hat{P}(s^{'} \vert s, \hat{a}) \tilde{h}^{*}(s^{'}) \right] \mathds{1}_{\hat{a} \neq \apass} \right\}. \label{eq:orig-atomic-policy-construct-h-body}
    \end{align}
\end{lemma}

Then, Lemma~\ref{lemma:opt-eq-atomic} shows that the relative value function $\tilde{h}^*$ coincides with the relative value function $h^*$ under the original MDP and that the corresponding deterministic step-independent atomic policy $\tilde{\pi}^*$ achieves the optimal long-run average reward $R(\pi^*)$.
The proofs of Lemmas \ref{lemma:atomic-h-exist}-\ref{lemma:atomic-mdp-equivalence} are provided in Appendix~\ref{sec:appendix-reduced-atomic}.

\begin{lemma} \label{lemma:opt-eq-atomic}
    For all $s \in S$, $\tilde{h}^{*}(s) = h^*(s)$. 
    Furthermore, the following deterministic atomic policy $\tilde{\pi}^{*}$ achieves the long-run average reward $R(\tilde{\pi}^{*})= \hat{g}^* = R(\pi^*)$:
    \begin{align}
        &\tilde{\pi}^{*}(\hat{a} \vert s) = \notag\\
        &\quad \begin{cases}
            1, \quad &\text{if } \hat{a} = \argmax_{\hat{a} \in \hat{\mathcal{A}}_{s}}\\
            &\quad\left\{ \left[r_H(s) -\hat{g}^* \right.\right.\\
            &\quad\quad \left.\left.+ \sum_{s' \in \mathcal{S}} P^{\text{sys}}(s' \vert s) \tilde{h}^{*}(s') \right]\mathds{1}_{\hat{a} = \apass} \right.\\
            &\quad + \left.\left[\hat{r}(\hat{a}) \right.\right.\\ 
            &\quad\quad \left.\left.- \sum_{s^{'} \in \mathcal{S}} \hat{P}(s^{'} \vert s, \hat{a}) \tilde{h}^{*}(s^{'}) \right] \mathds{1}_{\hat{a} \neq \apass} \right\},\\
            0, \quad &\text{otherwise},
        \end{cases} \notag\\ 
        &\quad \forall s \in \mathcal{S}. \label{eq:orig-atomic-policy-construct-pi}
    \end{align}
\end{lemma}

Lemma \ref{lemma:opt-eq-atomic} concludes the proof of Theorem \ref{thm:atomic-policy-equivalence}.
Thanks to Theorem \ref{thm:atomic-policy-equivalence}, we establish that the training can be restricted to step-independent atomic policies without loss of optimality, which further reduces the policy space. In the next step, we demonstrate how to obtain the optimal step-independent atomic policy using reinforcement learning approaches.


\section{Atomic-PPO Algorithm} \label{sec:atomic-ppo}
In this section, we present a deep reinforcement learning (RL) algorithm, referred as the {\em Atomic-PPO} (Algorithm \ref{algo:atomic-ppo-vanilla}), to train the optimal step-independent atomic policy. Atomic-PPO integrates our atomic action decomposition into the popular deep RL algorithm, Proximal Policy Optimization (PPO) \citep{schulman2017proximal}.

\begin{algorithm}[H]
    \SetAlgoLined
    \caption{Atomic-PPO} \label{algo:atomic-ppo-vanilla}
    \KwInputs{Number of policy iterations $N$, number of trajectories per policy iteration $M$, number of time steps $T$ in each trajectory, initial atomic policy function $\tilde{\pi}^{}_{\theta_0}$.}
    \For{policy iteration $\ell = 1, \dots, N$}{
        Run policy $\tilde{\pi}^{}_{\theta_{\ell - 1}}$ to collect $M$ trajectories, each has length $T$ time steps. Collect data as in \eqref{eq:atomic-ppo-data}. \\
        Estimate the advantage functions $\bar{A}_{k, \theta_{\ell-1}}(s^{t, m}_{k}, \hat{a}^{t, m}_{k})$ for each atomic step $k \in [K]$ at time $t \in [T]$ of trajectory $m \in [M]$, as in \eqref{eq:atomic-ppo-advantage}.\\ 
        Update the policy parameters by maximizing the surrogate objective \eqref{eq:atomic-ppo-obj} to obtain the new policy $\tilde{\pi}_{\theta_{\ell}}$.
    }
    \Return{atomic policy $\tilde{\pi}^{}_{\theta_{N}}$}
\end{algorithm}

We parameterize the step-independent atomic action policy $\tilde{\pi}$ using a neural network, and denote the parameterized policy as $\tilde{\pi}_{\theta}$, where $\theta$ represents the parameters of the neural network. We initialize the neural network parameters to be $\theta_0$. In each policy iteration $\ell = 1, \dots, N$, we collect a dataset $\mathrm{Data}_{\theta_{\ell-1}}^{M}$ by simulating $M$ trajectories using the atomic action policy $\tilde{\pi}^{}_{\theta_{\ell-1}}$ from the previous iteration, truncating each trajectory to $T$ time steps.
Here, we set $T$ to be sufficiently large so that the state distribution at time $T$ is close to the stationary state distribution, regardless of the initial state. The dataset collected at each time $t \in [T]$ of trajectory $m \in [M]$ includes the sequence of atomic actions $(\hat{a}^{t,m}_k)_{k \in [K]}$ generated by the policy and the induced sequence of states $(s^{t,m}_k)_{k \in [K+1]}$.
\begin{align} \label{eq:atomic-ppo-data}
    &\mathrm{Data}_{\theta_{\ell-1}}^{M} := \notag\\
    &\left\{ \left[(s^{t,m}_k)_{k \in [K+1]}, (\hat{a}^{t,m}_k)_{k \in [K]} \right]_{t \in [T]} \right\}_{m \in [M]}.
\end{align}

Using the collected data, we estimate the following quantities:
\begin{enumerate}
    \item[(i)] The empirical estimate $\bar{g}_{\theta_{\ell-1}}$ of the long-run average reward $g_{\theta_{\ell-1}} := R(\tilde{\pi}^{}_{\theta_{\ell-1}})$ for policy $\tilde{\pi}^{}_{\theta_{\ell-1}}$, computed as the total rewards averaged over $M$ trajectories:
    \begin{align}
        &\bar{g}_{\theta_{\ell-1}} = \frac{1}{MT} \sum_{m \in [M]} \sum_{t \in [T]} \notag \\
        &\quad \left(\sum_{k \in [K]} \hat{r}(\hat{a}^{t,m}_k) + r_H(s^{t,m}_{K+1}) \right).
    \end{align}
    \item[(ii)] The relative value function $\hat{h}_{k,\theta_{\ell-1}}: \mathcal{S} \rightarrow \mathbb{R}$ given policy $\tilde{\pi}_{\theta_{\ell-1}}$ at each atomic step $k \in [K]$ quantifies the difference between the expected total reward starting from a given state and the long-run average reward under the policy $\tilde{\pi}_{\theta_{\ell-1}}$:
    \begin{align} \label{eq:h-def-standard-limit-main}
        & \hat{h}_{k,\theta_{\ell-1}}(s) \notag\\
        =& \mathbb{E}_{\hat{\pi}_{\theta_{\ell-1}}} \left[\sum_{\ell = k }^K \left( \hat{r}(\hat{a}_{\ell}^1) - \frac{1}{K} \hat{g}_{\hat{\pi}} \right) + r_H(s^1_{K+1}) \Bigg\lvert s^1_k = s \right] \notag \\
        &\quad+ \sum_{t = 2}^{\infty} \mathbb{E}_{\hat{\pi}_{\theta_{\ell-1}}} \notag\\
        &\quad \left[\sum_{\ell = 1}^K \left( \hat{r}(\hat{a}_{\ell}^t) - \frac{1}{K} \hat{g}_{\hat{\pi}} \right) + r_H(s^t_{K+1}) \Bigg\lvert s^1_k = s \right], \notag\\
        &\quad\quad \forall s \in \mathcal{S}.
    \end{align}
    \begin{remark}
        In Appendix \ref{sec:appendix-h-def}, we establish in Proposition \ref{prop:h-def-well-defined} that the relative value function $\hat{h}_{k,\theta_{\ell-1}}$ defined by \eqref{eq:h-def-standard-limit-main} is well defined. Furthermore, in Proposition \ref{prop:h-def-equivalence}, we show that this definition is equivalent, up to an additive constant, to the standard definition of the relative value function for periodic chains based on the Cesàro limit, as presented in \cite{PutermanMDP}. We adopt the form in \eqref{eq:h-def-standard-limit-main} because it simplifies the empirical estimation, which we will detail next. Importantly, its equivalence with the standard definition ensures that it does not alter the training process.
    \end{remark}
    
    We compute the empirical estimate of the relative value function using the $TD(\lambda)$ approach, which has the convergence guarantee to the true relative value function \citep{tsitsiklis1999average}.
    Specifically, for each atomic step $k \in [K]$, we approximate the relative value function $\hat{h}_{k, \theta_{\ell-1}}$ with a neural network $\bar{h}_{k, \psi_{\ell-1}}: \mathcal{S} \to \mathbb{R}$, where $\psi_{\ell-1}$ denotes the neural network parameters. At each state $s_k^{t,m}$ of the $k$-th atomic step at time $t \in [T]$ in trajectory $m \in [M]$, we construct an empirical $TD(\lambda)$ estimate, denoted $\bar{h}_{k,\theta_{\ell-1}}^{t,m}$, to approximate $\hat{h}_{k,\theta_{\ell-1}}(s_k^{t,m})$:
    \begin{align} \label{eq:h-estimate}
        &\bar{h}_{k,\theta_{\ell-1}}^{t,m} \notag \\
        =& (1-\lambda) \sum_{j = t}^T \lambda^{j-t} \cdot \notag\\ 
        &\quad \left\{\sum_{\ell = k}^K \left(\hat{r}(\hat{a}^{t,m}_{\ell}) - \frac{1}{K}\bar{g}_{\theta_{\ell-1}} \right) \right. + r_H(s_{K+1}^{t,m})  \notag\\
        &+ \sum_{t' = t+1}^j \left[\sum_{\ell = 1}^K \left(\hat{r}(\hat{a}^{t',m}_{\ell}) - \frac{1}{K}\bar{g}_{\theta_{\ell-1}} \right) + r_H(s_{K+1}^{t',m}) \right] \notag \\
        &\quad \left. +\bar{h}_{1,\psi_{\ell-2}}(s^{j+1,m}_1) \right\}.
    \end{align}
    In \eqref{eq:h-estimate}, the inner expression for each $j$ is an empirical approximation of the value defined in \eqref{eq:h-def-standard-limit-main}, truncated at horizon $j$. The sums beyond $j+1$ is estimated using the neural network $\bar{h}_{1,\psi_{\ell-2}}$ obtained from the previous policy iteration, where the network output $\bar{h}_{1,\psi_{-1}}(s^{j+1,m}_1)$ is initialized randomly for the initial policy iteration. The estimate $\bar{h}_{k,\theta_{\ell-1}}^{t,m}$ is then computed by exponentially averaging these approximations across different horizons $j \geq t$.
    
    We learn the parameterized function $\bar{h}_{k, \psi_{\ell-1}}$ by minimizing the mean-square loss given the empirical TD($\lambda$) estimates:
    \begin{align}
        \sum_{m \in [M]} \sum_{t \in [T]} \left(\bar{h}_{k, \psi_{\ell-1}}(s^{t,m}_k) - \bar{h}^{t,m}_{k,\theta_{\ell-1}} \right)^2.
    \end{align}
    \item[(iii)] The advantage functions $\hat{A}_{\theta_{\ell-1}} := \{\hat{A}_{k,\theta_{\ell-1}}: \mathcal{S} \times \hat{\mathcal{A}} \rightarrow \mathbb{R}\}_{k \in [K]}$ of $\tilde{\pi}^{}_{\theta_{\ell-1}}$: The advantage function $\hat{A}_{k,\theta_{\ell-1}}(s,\hat{a})$ quantifies the changes in the long-run average reward by selecting atomic action $\hat{a} \in \hat{\mathcal{A}}$ at state $s \in \mathcal{S}$ during atomic step $k \in [K]$ instead of following the atomic policy $\tilde{\pi}_{\theta_{\ell-1}}$, while assuming $\tilde{\pi}_{\theta_{\ell-1}}$ is followed for all subsequent steps. In particular, the advantage function $\hat{A}_{\theta_{\ell-1}}$ for the atomic policy $\tilde{\pi}_{\theta_{\ell-1}}$ is defined as 
    \begin{align*}
        &\hat{A}_{k, \theta_{\ell-1}}(s, \hat{a}) \\
        =& \begin{cases}
            &\hat{r}(\hat{a}) - \frac{1}{K} \hat{g}_{\theta_{\ell-1}} + \\
            &\quad \sum_{s' \in \mathcal{S}} \hat{P}(s' \vert s, \hat{a}) \hat{h}_{k+1, \theta_{\ell-1}}(s') - \hat{h}_{k, \theta_{\ell-1}}(s),\\
            &\quad\quad \text{if } k < K,\\
            &\hat{r}(\hat{a}) - \frac{1}{K} \hat{g}_{\theta_{\ell-1}} + \sum_{s' \in \mathcal{S}} \hat{P}(s'\vert s, \hat{a}) r_H(s') + \\
            &\quad \sum_{s',s'' \in \mathcal{S}} P^{\text{sys}}(s'' \vert s') \hat{P}(s' \vert s, \hat{a}) \hat{h}_{1, \theta_{\ell-1}}(s'') -\\
            &\quad \hat{h}_{k, \theta_{\ell-1}}(s),\\
            &\quad\quad \text{if } k = K,\\
        \end{cases}\\
        &\quad \forall s \in \mathcal{S},\ \hat{a} \in \hat{\mathcal{A}}.
    \end{align*}
    For every pair of state $s^{t,m}_k$ and atomic action $\hat{a}^{t,m}_k$ in the collected dataset $\mathrm{Data}_{\theta_{\ell-1}}^{M}$, we can compute its empirical estimates of the advantage function by using the empirical estimate $\bar{h}_{k,\psi_{\ell-1}}$ in place of $\hat{h}_{k}$:
    \begin{align} \label{eq:atomic-ppo-advantage}
        &\bar{A}_{k, \theta_{\ell-1}}(s^{t,m}_k, \hat{a}^{t,m}_k) = \notag\\
        &\begin{cases}
            \hat{r}(\hat{a}^{t,m}_k) - \frac{1}{K} \bar{g}_{\theta_{\ell-1}} + \\
            \quad \hat{h}_{k+1, \psi_{\ell-1}}(s^{t,m}_{k+1}) - \hat{h}_{k, \psi_{\ell-1}}(s^{t,m}_k), \\
            \quad\quad \text{if } k < K, \\
            \hat{r}(\hat{a}^{t,m}_k) - \frac{1}{K} \bar{g}_{\theta_{\ell-1}} + r_H(s^{t,m}_{k+1}) + \\
            \quad \hat{h}_{1, \psi_{\ell-1}}(s^{t+1,m}_{1}) - \hat{h}_{k, \psi_{\ell-1}}(s^{t,m}_k), \\
            \quad\quad \text{if } k = K, \\
        \end{cases}
    \end{align}
\end{enumerate}

Using the estimated advantage function $\bar{A}_{\theta_{\ell-1}}$, the Atomic-PPO algorithm computes the atomic action policy function $\tilde{\pi}^{}_{\theta_{\ell}}$ of the next iteration by choosing parameter $\theta_{\ell}$ that maximizes the clipped surrogate function introduced by \cite{schulman2017proximal}:
\begin{align} \label{eq:atomic-ppo-obj}
    &\bar{L}(\theta_{\ell}, \theta_{\ell-1}) := \frac{1}{M}\sum_{m, t, k} \notag\\
    &\min\left(\frac{\tilde{\pi}^{}_{\theta_{\ell}}(\hat{a}^{t, m}_{k} \vert s^{t, m}_{k})}{\tilde{\pi}^{}_{\theta_{\ell-1}}(\hat{a}^{t, m}_{k} \vert s^{t, m}_{k})} \bar{A}_{k, \theta_{\ell-1}}(s^{t, m}_{k}, \hat{a}^{t, m}_{k}), \right. \nonumber \\
    & \text{clip} \left(\frac{\tilde{\pi}^{}_{\theta_{\ell}}(\hat{a}^{t, m}_{k} \vert s^{t, m}_{k})}{\tilde{\pi}^{}_{\theta_{\ell-1}}(\hat{a}^{t, m}_{k} \vert s^{t, m}_{k})}, 1 - \epsilon, 1 + \epsilon \right) \cdot \notag\\
    &\quad \left. \bar{A}_{k, \theta_{\ell-1}}(s^{t, m}_{k}, \hat{a}^{t, m}_{k}) \right),
\end{align}
where $\epsilon \in (0, 1)$ is the clip size.
The policy iteration repeats until $N$ iterations are completed, and the final atomic policy $\tilde{\pi}_{\theta_N}$ is returned. The complete procedure is summarized in Algorithm~\ref{algo:atomic-ppo-vanilla}.

\section{Applications} \label{sec:applications}

In this section, we illustrate the versatility of our stochastic processing network model from Section~\ref{sec:model} by showing how it captures three important application domains: hospital inpatient overflow assignment, internet switch scheduling, and ride-hailing. These examples are drawn from prior studies (\cite{sun2024inpatient, Huo_Switch_Scheduling_via, dai2025atomicEV, feng2021scalable}), where Atomic-PPO has been applied and demonstrated strong performance. We summarize the findings of these studies to show concretely how these works map into our SPN model classes and how the atomic decomposition framework can be applied.


\subsection{Hospital Inpatient Overflow Assignment}
We adopt the hospital system model from \cite{sun2024inpatient} (illustrated in Figure 3 of their paper). The model considers $I$ patient classes (item classes) and $I$ pools (i.e. inpatient units). Patient arrivals occur at each discrete time $t \in [T]$ (e.g., hourly). Each arriving patient is placed on a wait list corresponding to their class until assigned to a unit. Each unit $i \in \mathcal{I} := \{1, \dots, I\}$ is specialized in treating patients of class $i$ and contains $K_i$ identical beds, resulting in a total of $K = \sum_{i \in \mathcal{I}} K_i$ beds across the system. Units can also admit patients from other classes at an associated {\em overflow cost}. Patients are discharged from the hospital once they have been served by the unit.

We model each undischarged patient as an {\em item} and each patient class as an {\em item class}. We model each bed as a {\em server}. The set of {\em service types} is defined as $\mathcal{J} := \{(i, i') : i, i' \in \mathcal{I} \}$, where $(i, i')$ denotes assigning a bed from unit $i$ to serve a patient of class $i'$. There are $J = I^2$ total service types.
To match our SPN formulation in Section \ref{sec:model}, we express all system costs as negative rewards: 
\begin{enumerate}
    \item[(i)] {\em Overflow cost}: Assigning a patient of class $i' \in \mathcal{I}$ to an inpatient unit $i \in \mathcal{I}$ (where $i \neq i'$) incurs a one-time reward $r_{(i,i')} < 0$.
    \item[(ii)] {\em Holding cost}: Patients remaining on the wait lists will incur a per unit-time cost (i.e. negative reward). 
    Given a post-decision state $s^{t_+} := (z^{t_+}, n^{t_+})$ at time $t$, we can obtain the number of patients of each class $i \in \mathcal{I}$ on wait lists by subtracting the total number of class $i$ patients currently being served in all inpatient units, given by $\sum_{i' \in \mathcal{I}} \sum_{\tau \in \mathbb{N}_+} n_{(i', i), \tau}^{t_+}$, from the total number of class $i$ patients in the system $z_i^{t_+}$.
\end{enumerate}

At each hour, the hospital manager selects a {\em schedule}, which is an action that assigns a service to every empty bed across all units. 
By modeling the assignment of each individual empty bed as an atomic action, we can reduce the action size from being exponential in the number of empty beds to being a constant.

Sun et al. tested Atomic-PPO on hospital systems with 5, 10, and 20 inpatient units and compared the Atomic-PPO to several benchmark policies. Their atomic action decomposition is similar to ours, with the main difference being how probability ratios and advantage functions are defined: their version uses the original actions, while ours uses atomic actions directly (see Equation~\eqref{eq:atomic-ppo-obj}).
In the case with 5 inpatient units, they compared Atomic-PPO to an approximate dynamic programming (ADP) method, which was a state-of-art algorithm for inpatient assignment. Atomic-PPO performed similarly to ADP (within 4\%, see Table 3 in \cite{sun2024inpatient}) but required much less training time, with approximately 2 hours per policy iteration compared to more than 10 hours for ADP.
For larger systems with 10 or 20 inpatient units, ADP becomes intractable because it requires to search over the combinatorially large action space. In these cases, Atomic-PPO was compared to three benchmark policies from \cite{dai2019inpatient}.
As shown in Figure 4 of \cite{sun2024inpatient}, Atomic-PPO outperformed all three benchmarks and reduced long-run average costs by about 25\% in both the 10-unit and 20-unit cases.

\subsection{Switch Scheduling}
Input-queued switch is a key component in high-speed communication networks that routes packets from multiple input ports to multiple output ports. We consider the input-queued switch model from \cite{Huo_Switch_Scheduling_via} (illustrated in Figure 3.3 of their paper), where a $W \times W$ crossbar switch connects $W$ input ports to $W$ output ports. Each input maintains $W$ virtual output queues (VOQs), one for each output, to buffer packets awaiting transmission. The VOQ at input $w'$ for output $w$ (denoted $\text{VOQ}_{w'w}$) holds packets of class $(w', w)$, representing those that arrive at input $w'$ and are destined for output $w$. Packets of each class arrive independently at every time step. Each input port can send a packet to at most one output port, and each output port can receive a packet from at most one input port. Transmissions take one time step, incur no cost or reward, and packets depart the system immediately upon transmission. Each queued packet incurs a holding cost of 1 per time step. At each time step, the system manager selects a one-to-one matching between input and output ports, determining which packets to transmit. The objective is to minimize the long-run average holding cost.

We map this system into the SPN model in Section \ref{sec:model} by representing each output port as a {\em server}, yielding a total of $K = W$ servers. Each packet is an {\em item}, and packets in $\text{VOQ}_{w'w}$ belong to {\em item class} $(w', w)$. The set of item classes is $\mathcal{I} := \{(w', w) : w', w \in [W]\}$, with $I = W^2$ total classes. Each {\em service type} specifies how a server (i.e., an output port) is used during a time step by defining the input–output port pair to be connected for packet transmission. 
A {\em schedule} (i.e., action) is a matching between input and output ports, with an action space of size $K!$. On the other hand, each atomic action only consists of assigning one output port to one input port, yielding an atomic action space of size $K^2$.

Huo et al. considered the internet switch scheduling problem with three different packet arrival patterns: the uniform arrival, the diagonal arrival, and the bottom-skewed arrival. They evaluated the performance of Atomic-PPO using 5, 6, 7, and 8 switches. In these cases, Atomic-PPO was compared to three benchmark policies: Max Weight, $d$-flip, and greedy. As shown in Table 3.4-3.6 of \cite{Huo_Switch_Scheduling_via}, Atomic-PPO significantly outperformed both the $d$-flip and greedy policies, and achieved performance comparable to the Max Weight policy. 

\subsection{Ride-Hailing}
We adopt the transportation network model from \cite{dai2025atomicEV}, which considers a set of regions $\mathcal{W}$ and a fleet of $K$ electric vehicles (EVs) serving passenger trip requests. Each vehicle has a range limit. In each region, there are chargers of various outlet rates. 
Trip requests from origin $w'$ to destination $w$ arrive at each minute $t \in [T]$ of the day. Each trip has a fixed duration and battery cost. 
At each minute, the system manager can assign vehicles to one of three tasks: (i) fulfilling passenger trips, (ii) repositioning to a different region, or (iii) charging. The system earns rewards from fulfilling trips and incurs costs from repositioning and charging. 


We model each vehicle as a {\em server} and each trip request as an {\em item}. Each {\em item class} corresponds to an origin-destination pair $(w', w)$. There are three {\em types of services}: (i) Trip-fulfillment services, which match vehicles (with specific location, battery, and availability) to trip requests; (ii) Repositioning services, which move idle vehicles to another region; (iii) Charging services, which recharge idle vehicles.
At each minute, the system manager selects a {\em schedule} (i.e., action) that assigns each of the $K$ vehicles to some service. By modeling the assignment of each individual vehicle as an atomic action, we reduce the action size from being exponential in the number of vehicles to being a constant.

Dai et al. study a ride-hailing system in the Manhattan area of New York City, where the service region is partitioned into 10 zones. The paper considers trip requests on workdays, with the fleet size set to 300 vehicles.
They train the fleet dispatching policy using Atomic-PPO, implemented as described in Section~\ref{sec:atomic-ppo}. Each policy iteration simulates 30 trajectories of 8 workdays and takes 15–20 minutes with parallel processing on 30 CPUs. The algorithm converges within 10 iterations, completing training in under 3 hours.
To evaluate performance, they derive a fluid-based upper bound and compare Atomic-PPO with two benchmarks: the power-of-$m$ policy, which assigns each trip to the highest-battery vehicle among the $m$ closest; and the fluid policy, which applies randomized rounding to the fluid LP solution. As shown in Table 4 of \cite{dai2025atomicEV}, Atomic-PPO achieves 91\% of the upper bound, significantly outperforming the power-of-$m$ policy (71\%) and the fluid policy (43\%).

Feng et al. \citep{feng2021scalable} study a ride-hailing system without incorporating the charging aspect, under a finite time horizon. They evaluate the performance of Atomic-PPO in two settings: (i) a five-region system with 1,000 vehicles over a 360-minute horizon, and (ii) a nine-region system with 2,000 vehicles over a 240-minute horizon. The benchmark is the {\em time-dependent lookahead policy} proposed in \cite{braverman2019empty}, a heuristic fleet control strategy derived from fluid-based analysis. In both settings, Atomic-PPO outperforms the benchmark by 2–3\%.

\section{Concluding remarks} \label{sec:conclusion}
In this article, we propose the atomic action decomposition framework, which enables the efficient computation of optimal control policies in stochastic processing networks with a large number of servers. By exploiting the reward decomposability property, which holds broadly across many applications, our approach reduces the action size from being exponential in the number of servers to a constant, without loss of optimality. Furthermore, we introduce Atomic-PPO, an efficient reinforcement learning algorithm that integrates this framework into the original widely used PPO algorithm. We demonstrate that our model can be applied to hospital inpatient overflow assignment, switch scheduling, and ride-hailing, where atomic action decomposition substantially reduces the complexity of policy training and Atomic-PPO achieves strong performance. One valuable direction of future research is to explore how atomic action decomposition can be adapted to systems with non-additive and more complex reward structures.



%
%
%
\clearpage

\appendix
\section{Proof of Theorem \ref{thm:atomic-policy-equivalence}} \label{sec:appendix-reduced-atomic}

Before proving Theorem~\ref{thm:atomic-policy-equivalence}, we first prove Lemma~\ref{lemma:atomic-mdp-equivalence}, which establishes that the passing-last atomic MDP is equivalent to the $K$-step atomic MDP (defined in Section~\ref{sec:atomic-decomp}) under deterministic step-independent atomic policies.



\begin{lemma} \label{lemma:atomic-mdp-equivalence}
    For any deterministic step-independent atomic policy $\tilde{\pi}$, the long-run average reward of the passing-last atomic MDP is equal to that with the $K$-step atomic MDP.
\end{lemma}
\begin{proof}{Proof of Lemma \ref{lemma:atomic-mdp-equivalence}}
    Consider any state $s_1^t \in \mathcal{S}$ at time $t$, before any atomic actions have been generated. We roll out both atomic MDPs using policy $\tilde{\pi}$, and let $k \in [K]$ be the atomic step at which $\tilde{\pi}$ first generates the passing atomic action $\apass$. Under both atomic MDPs, the system generates the same atomic rewards and follows the identical transition dynamics given by \eqref{eq:atomic-item-transition}–\eqref{eq:atomic-service-transition-inf} for the first $k-1$ atomic steps. Then, under the passing-last atomic MDP, the system terminates the current time step, with the post-decision state being $s_k^t$. Under the $K$-step atomic MDP, the system obtains an atomic reward of $0$ and the state remains unchanged, i.e., $s_{k+1}^t = s_k^t$. Since $\tilde{\pi}$ is deterministic, it will generate $\apass$ at the atomic step $k+1$ and will continue generating $\apass$ until all $K$ atomic steps are complete, with the post-decision state being $s_{K+1}^t = s_k^t$.
    Therefore, $\tilde{\pi}$ induces the same post-decision state and total reward at every time step in both atomic MDPs, and consequently achieves the same long-run average reward.
    \hfill $\square$
\end{proof}

As a consequence of Lemma \ref{lemma:atomic-mdp-equivalence}, the long-run average reward $R(\tilde{\pi}^{})$ achieved by $\tilde{\pi}^{}$ can be equivalently represented as: \[
    R(\tilde{\pi}^{}) = \lim_{T \to \infty} \frac{1}{T} \mathbb{E}_{\tilde{\pi}^{}}\left[\sum_{t=1}^T \sum_{k=1}^{K^{t}} \hat{r}(\hat{a}_k^t) + r_H(s^t_{K^t+1}) \right],
\] where the random variable $K^t$ denotes the length of the sequence of {\em non-passing} atomic actions in time step $t$. We note that $K^t \leq K$, as each non-passing atomic action reduces the number of idling servers in the state space by one. Consequently, the number of non-passing atomic actions in any given time step is upper bounded by the total number of servers in the system.

\begin{proof}{Proof of Lemma \ref{lemma:atomic-h-exist}}

    For each $k = 0, \dots, K$, we use $\mathcal{S}[k]$ to denote the set of all states with $k$ idling servers. Therefore, we have $\cup_{k = 0}^K \mathcal{S}[k] = \mathcal{S}$. Moreover, for any $k,k' \in [K]$ with $k \neq k'$, we must have $\mathcal{S}[k] \cap \mathcal{S}[k'] = \varnothing$.
    
    We first consider the set of states in $\mathcal{S}[0]$, for which $\apass$ is the only feasible atomic action. Then, equation \eqref{eq:orig-atomic-policy-construct-h-body} indicates that $\tilde{h}^{*}(s) = r_H(s) -\hat{g}^* + \sum_{s' \in \mathcal{S}} P^{\text{sys}}(s' \vert s) h^*(s')$ for all $s \in \mathcal{S}[0]$. Now, consider the set of states in $\mathcal{S}[1]$, where we can assign at most one non-passing atomic action. For any such state $s \in \mathcal{S}[1]$, \eqref{eq:orig-atomic-policy-construct-h-body} indicates that $\tilde{h}^{*}$ can be computed by
    \begin{align*} 
        &\tilde{h}^{*}(s) = \\
        &\max_{\hat{a} \in \hat{\mathcal{A}}_{s}}\left\{ \left[r_H(s) -\hat{g}^* + \sum_{s' \in \mathcal{S}} P^{\text{sys}}(s' \vert s) h^*(s') \right] \mathds{1}_{\hat{a} = \apass} \right.\\ 
        &\quad+ \left. \left[\hat{r}(\hat{a}) - \sum_{s^{'} \in \mathcal{S}[0]} \hat{P}(s^{'} \vert s, \hat{a}) \tilde{h}^{*}(s^{'}) \right] \mathds{1}_{\hat{a} \neq \apass} \right\},
    \end{align*}
    where the non-passing action leads to a state $s^{'} \in \mathcal{S}[0]$ such that $\tilde{h}^{*}$ value is already computed. 
    
    We note that this procedure can be repeated for $k=2, 3, \dots, K$. 
    For state $s \in \mathcal{S}[k]$ that can be assigned at most $k$ non-passing atomic actions, \eqref{eq:orig-atomic-policy-construct-h-body} indicates that 
    \begin{align*} 
        &\tilde{h}^{*}(s) = \\
        &\max_{\hat{a} \in \hat{\mathcal{A}}_{s}}\left\{\left[r_H(s) -\hat{g}^* + \sum_{s' \in \mathcal{S}} P^{\text{sys}}(s' \vert s) h^*(s') \right]\mathds{1}_{\hat{a} = \apass} \right.\\
        &\quad+ \left. \left[\hat{r}(\hat{a}) - \sum_{s^{'} \in \mathcal{S}[k-1]} \hat{P}(s^{'} \vert s, \hat{a}) \tilde{h}^{*}(s^{'}) \right] \mathds{1}_{\hat{a} \neq \apass} \right\},
    \end{align*}
    where the non-passing action leads to a state in $\mathcal{S}[k-1]$, for which the corresponding $\tilde{h}^{*}$ value has already been computed. We repeat this process until $k = K$ and therefore we can obtain the $\tilde{h}^{*}$ value for all $s \in \mathcal{S}$.
    \hfill $\square$
\end{proof}

\begin{proof}{Proof of Lemma \ref{lemma:opt-eq-atomic}}
    Now, we are going to show that the constructed atomic policy $\tilde{\pi}^{*}$ in \eqref{eq:orig-atomic-policy-construct-pi} achieves the same long-run average reward as $\pi^*$ by first proving that $\tilde{h}^{*}(s) = h^*(s)$ for all $s \in \mathcal{S}$.
    Given any state $s \in \mathcal{S}$, let $(\hat{a}_k)_{k = 1}^{\tilde{K}^{\dagger}+1}$ be the sequence of atomic actions generated by $\tilde{\pi}^{*}$, where $\hat{a}_{\tilde{K}^{\dagger}+1} = \apass$ for some $\tilde{K}^{\dagger} \leq K$. Let $(s_k)_{k=1}^{\tilde{K}^{\dagger}+1}$ be the corresponding sequence of states induced by these atomic actions. 
    Let $a^{\dagger}$ be the action that is {\em equivalent} to $(\hat{a}^{\dagger}_k)_{k = 1}^{\tilde{K}^{\dagger}+1}$ in that $a^{\dagger}_{j'j} = \sum_{k = 1}^{\tilde{K}^{\dagger}} \hat{a}^{\dagger}_{j'j},\  \forall j,j' \in \mathcal{J}$. We first argue that $h^*(s) \geq \tilde{h}^{*}(s)$:
    \begin{subequations} \label{eq:h-hat-leq-h}
        \begin{align}
            h^{*}(s) =& \max_{a \in \mathcal{A}_{s}}\left\{r(s, a) -\hat{g}^* + \sum_{s' \in \mathcal{S}} P(s' \vert s, a) h^{*}(s') \right\} \notag \\
            \geq& r(s, a^{\dagger}) -\hat{g}^* + \sum_{s' \in \mathcal{S}} P(s' \vert s, a^{\dagger}) h^{*}(s') \nonumber\\
            =& \sum_{k = 1}^{\tilde{K}^{\dagger}} \hat{r}(\hat{a}^{\dagger}_k) + r_H(s_{\tilde{K}^{\dagger}+1}) -\hat{g}^* \notag \\
            &\quad+ \sum_{s' \in \mathcal{S}} P^{\text{sys}}(s' \vert s_{\tilde{K}^{\dagger}+1}) \cdot \notag\\
            &\quad \hat{P}(s_{\tilde{K}^{\dagger}+1} \vert s_{\tilde{K}^{\dagger}}, \hat{a}^{\dagger}_{\tilde{K}^{\dagger}}) \cdots \hat{P}(s_2 \vert s_1, \hat{a}^{\dagger}_1) h^{*}(s') \label{eq:opt-eq-atomic-proof-left-a}\\
            =& \sum_{k = 1}^{\tilde{K}^{\dagger}} \hat{r}(\hat{a}_k^{\dagger}) + r_H(s_{\tilde{K}^{\dagger}+1}) -\hat{g}^* \notag \\
            &\quad+ \sum_{s' \in \mathcal{S}} P^{\text{sys}}(s' \vert s_{\tilde{K}^{\dagger}+1}) h^{*}(s') \label{eq:opt-eq-atomic-proof-left-b} \\
            =& \tilde{h}^{*}(s), \label{eq:opt-eq-atomic-proof-left-c}
        \end{align}
    \end{subequations}
    where \eqref{eq:opt-eq-atomic-proof-left-a} is obtained due to the fact that $a^{\dagger}$ and $(\hat{a}^{\dagger}_k)_{k=1}^{\tilde{K}^{\dagger}+1}$ are equivalent; \eqref{eq:opt-eq-atomic-proof-left-b} is because the state transition given non-passing action is deterministic, i.e. $\hat{P}(s_{\tilde{K}^{\dagger}+1} \vert s_{\tilde{K}^{\dagger}}, \hat{a}^{\dagger}_{\tilde{K}^{\dagger}}) = \dots = \hat{P}(s_2 \vert s_1, \hat{a}^{\dagger}_1) = 1$; and \eqref{eq:opt-eq-atomic-proof-left-c} is obtained by applying \eqref{eq:orig-atomic-policy-construct-h-body} recursively for $\tilde{K}^{\dagger}+1$ steps using atomic actions generated by $\tilde{\pi}^{*}$ that satisfies \eqref{eq:orig-atomic-policy-construct-pi}. 
    
    We next show that $h^*(s) \leq \tilde{h}^{*}(s)$. Let $a^{*}$ be the action returned by $\pi^*$ at state $s$. Let $(\hat{a}^{*}_{k})_{k=1}^{\tilde{K}^{*}+1}$ be the sequence of atomic actions that is equivalent to $a^{*}$, where $\hat{a}^{*}_{\tilde{K}^{*}+1} = \apass$ for some $\tilde{K}^{*} \leq K$. We then obtain 
    \begin{subequations} \label{eq:h-hat-geq-h}
        \begin{align} 
            \tilde{h}^{*}(s) \geq& \sum_{k = 1}^{\tilde{K}^{*}} \hat{r}(\hat{a}^{*}_k) + r_H(s_{\tilde{K}^{*}+1}) -\hat{g}^* \notag\\
            &\quad+ \sum_{s^{*} \in \mathcal{S}} P^{\text{sys}}(s' \vert s_{\tilde{K}^{*}+1}) h^{*}(s') \label{eq:h-hat-geq-h-a}\\
            =& \sum_{k = 1}^{\tilde{K}^{'}} \hat{r}(\hat{a}^{*}_k) + r_H(s_{\tilde{K}^{*} + 1}) -\hat{g}^* \notag\\
            &\quad+ \sum_{s' \in \mathcal{S}} P^{\text{sys}}(s' \vert s_{\tilde{K}^{*} + 1}) \cdot \notag\\
            &\quad \hat{P}(s_{\tilde{K}^{*}+1} \vert s_{\tilde{K}^{*}}, \hat{a}^{*}_{\tilde{K}^{*}}) \cdots \hat{P}(s_2 \vert s_1, \hat{a}^{*}_1) h^{*}(s') \label{eq:h-hat-geq-h-b}\\
            =& r(s, a^{*}) -\hat{g}^* + \sum_{s' \in \mathcal{S}} P(s' \vert s_1, a^{*}) h^{*}(s') \label{eq:h-hat-geq-h-c}\\
            =& h^*(s), \label{eq:h-hat-geq-h-d}
        \end{align}
    \end{subequations}
    where \eqref{eq:h-hat-geq-h-a} is obtained by applying \eqref{eq:orig-atomic-policy-construct-h-body} recursively for $\tilde{K}^{*}+1$ steps using atomic actions generated by $\tilde{\pi}^{*}$; \eqref{eq:h-hat-geq-h-b} is because the state transitions given non-passing atomic actions are deterministic, i.e. $\hat{P}(s_{\tilde{K}^{*}+1} \vert s_{\tilde{K}^{*}}, \hat{a}^{*}_{\tilde{K}^{*}}) = \dots = \hat{P}(s_2 \vert s_1, \hat{a}^{*}_1) = 1$; \eqref{eq:h-hat-geq-h-c} is due to the fact that $a^{*}$ is equivalent to $(\hat{a}_1^{*}, \dots, \hat{a}_{\tilde{K}^{*}+1}^{*})$; and \eqref{eq:h-hat-geq-h-d} is due to the fact that $a^{*}$ is generated by the optimal policy $\pi^*$ that satisfies \eqref{eq:orig-atomic-policy-construct-pi-aug}.

    Using \eqref{eq:h-hat-leq-h} and \eqref{eq:h-hat-geq-h}, we obtain $\tilde{h}^{*}(s) = h^*(s)$ for all state $s$. Hence, we obtain that:
    \begin{align} 
        &\tilde{h}^{*}(s) = \notag\\
        &\max_{\hat{a} \in \hat{\mathcal{A}}_{s}} \left\{ \left[r_H(s) -\hat{g}^* + \sum_{s' \in \mathcal{S}} P^{\text{sys}}(s' \vert s) \tilde{h}^{*}(s') \right]\mathds{1}_{\hat{a} = \apass} \right. \notag\\
        &\quad+ \left. \left[\hat{r}(\hat{a}) + \sum_{s^{'} \in \mathcal{S}} \hat{P}(s^{'} \vert s, \hat{a}) \tilde{h}^{*}(s^{'}) \right] \mathds{1}_{\hat{a} \neq \apass}\right\}, \quad \forall s \in \mathcal{S}. \label{eq:orig-atomic-policy-construct-h-body-new}
    \end{align}
    
    Now we are ready to show that $R(\tilde{\pi}^{*}) = R^*$. Consider any state $s \in \mathcal{S}$. Let $\hat{a}_1, \dots, \hat{a}_{\tilde{K}^{*}+1}$ be the sequence of atomic actions generated by $\tilde{\pi}^{*}$ at the current time step, where $\hat{a}_{\tilde{K}^{*}+1} = \apass$ and $\tilde{K}^{*} \leq K$. Let $s_1, \dots, s_{\tilde{K}^{*}+1}$ be the corresponding sequence of states induced by these atomic actions. We can obtain 
    \begin{align*}
        \tilde{h}^{*}(s) =& \sum_{k = 1}^{\tilde{K}^{*}} \hat{r}(\hat{a}^{}_k) + r_H(s_{\tilde{K}^{*}+1}) - \hat{g}^* \\
        &\quad+ \sum_{s' \in \mathcal{S}} P^{\text{sys}}(s' \vert s_{\tilde{K}^{*}+1}) \tilde{h}^{*}(s'),
    \end{align*}
    Let $s \in \mathcal{S}$ be the initial state. Repeating the equation above for every time step $t = 1, \dots, T$, we obtain 
    \begin{align*}
        \tilde{h}^{*}(s) =& \mathbb{E}_{\tilde{\pi}^{*}} \left[\sum_{t = 1}^{T} \sum_{k = 1}^{K^{t}} \hat{r}(\hat{a}^{t}_k) + r_H(s^t_{K^t+1}) \right] \\
        &\quad- T \hat{g}^* + \mathbb{E}_{\tilde{\pi}^{*}} \left[\tilde{h}^{*}(s^{T+1}) \Bigg\lvert s \right].
    \end{align*}

    We divide both sides of the equation by $T$ and take the limit $T \rightarrow \infty$. Since $h^*$ is bounded, we have \[
        \hat{g}^* = \lim_{T \rightarrow \infty} \frac{1}{T} \mathbb{E}_{\tilde{\pi}^{*}} \left[\sum_{t = 1}^{T} \sum_{k = 1}^{K^{t}} \hat{r}(\hat{a}^{t}_k) + r_H(s^t_{K^t+1}) \right] = R(\tilde{\pi}^{*}).
    \]
    
    Recall from Theorem \ref{thm:atomic-optimal-aug} that $\hat{g}^* = R(\pi^*)$. Hence, we obtain that $R(\tilde{\pi}^{*}) = R(\pi^*)$, where $\tilde{\pi}^{*}$ is constructed by \eqref{eq:orig-atomic-policy-construct-pi}. 
    \hfill $\square$
\end{proof}

\section{Relative Value Function of Atomic Action Policies} \label{sec:appendix-h-def}
Consider any atomic action policy $\hat{\pi}:= \{\hat{\pi}_k\}_{k \in [K]}$. We use $\hat{h}_{k,\hat{\pi}}: \mathcal{S} \rightarrow \mathbb{R}$ to denote the relative value function of $\hat{\pi}$ at each atomic step $k \in [K]$. 
In Proposition \ref{prop:h-def-well-defined}, we demonstrate that the infinite series in \eqref{eq:h-def-standard-limit-main} is well defined. Then, in Proposition \ref{prop:h-def-equivalence}, we show that our definition of relative value function is equivalent to the relative value function defined using Cesaro limit, as given on page 338 of \cite{PutermanMDP} for periodic chains, up to an additive constant. 

\begin{proposition} \label{prop:h-def-well-defined}
    The function $\hat{h}_{k,\hat{\pi}}(s)$ given in \eqref{eq:h-def-standard-limit-main} is well defined for all $s \in \mathcal{S}$.
\end{proposition}

\begin{proposition} \label{prop:h-def-equivalence}
    For any atomic policy $\hat{\pi}$, there exists a scalar $\kappa_{\hat{\pi}}$ such that the following equation holds:
    \begin{align}
        &\hat{h}_{k,\hat{\pi}}(s) + \kappa_{\hat{\pi}} \notag\\
        =& \mathbb{E}_{\hat{\pi}_{}}\left[\sum_{\ell = k }^K \left( \hat{r}(\hat{a}_{\ell}^1) - \frac{1}{K} \hat{g}_{\hat{\pi}} \right) + r_H(s_{K+1}^1) \Bigg\lvert s^1_k = s \right] \notag \\
        &+ \lim_{D \rightarrow \infty} \frac{1}{DK} \sum_{T = 1}^D \sum_{q = 1}^{K} \notag\\
        &\quad \mathbb{E}_{\hat{\pi}_{}}\left[\sum_{t = 2}^{T-1} \sum_{\ell = 1}^K \left( \hat{r}(\hat{a}_{\ell}^t) - \frac{1}{K} \hat{g}_{\hat{\pi}} \right) + r_H(s_{K+1}^t) \right. \notag\\
        &\quad+ \left. \sum_{\ell = 1}^{q} \left( \hat{r}(\hat{a}_{\ell}^{T}) - \frac{1}{K} \hat{g}_{\hat{\pi}} \right) + r_H(s_{K+1}^T) \Bigg\lvert s^1_k = s \right],\notag\\
        &\quad \forall s \in \mathcal{S}, \label{eq:h-def-cesaro-limit}
    \end{align}
    where the right-hand side of \eqref{eq:h-def-cesaro-limit} is the relative value function of $\hat{\pi}$ defined using the Cesaro limit as given by \cite{PutermanMDP} for periodic chains (see page 338 of \cite{PutermanMDP}).
\end{proposition}
\begin{proof}{Proof of Proposition \ref{prop:h-def-well-defined}}
    We consider the Markov chain formed by observing the system state at the $\ell$-th atomic step of each time step, for any $\ell \in [K]$. The transition matrix of this chain under policy $\hat{\pi}$, denoted by $\hat{P}_{\ell, \hat{\pi}}$, is given by
    \begin{align*}
        \hat{P}_{\ell, \hat{\pi}} := (\hat{P}_{\hat{\pi}})^{K+1-\ell} \cdot P^{\text{sys}} \cdot (\hat{P}_{\hat{\pi}})^{\ell-1},
    \end{align*}
    where $\hat{P}_{\hat{\pi}}$ is the transition matrix for a single atomic step under policy $\hat{\pi}$. Let $P_{\hat{\pi}} := (\hat{P}_{\hat{\pi}})^K P^{\text{sys}}$ be the state transition matrix induced by $\hat{\pi}$ in the original MDP. Since the original MDP is an aperiodic unichain, the state space is finite and the chain is irreducible, it follows from Theorem 1.8.3 (page 41 of \cite{norris1998markov}) that the matrix $((\hat{P}_{\hat{\pi}})^K P^{\text{sys}})^t = P^t_{\hat{\pi}}$ converges to a stochastic matrix $P^*_{\hat{\pi}}$ as $t \rightarrow \infty$. Each row of $P^*_{\hat{\pi}}$ equals the stationary distribution $\hat{\rho}_{1, \hat{\pi}}$ of the original MDP, i.e. $P^*_{\hat{\pi}} = e \hat{\rho}_{1, \hat{\pi}}^T$, where $e$ is the column vector of all ones. Therefore, for each atomic step $\ell \in [K]$, the stationary distribution of the Markov chain formed by observing the state at $\ell$-th atomic step of each time step is given by \[
        \hat{\rho}_{\ell, \hat{\pi}}^T := \hat{\rho}_{1, \hat{\pi}}^T  (\hat{P}_{\hat{\pi}})^{\ell-1}.
    \] The long-run average reward $\hat{g}_{\ell, \hat{\pi}}$ achieved by $\hat{\pi}$ at atomic step $\ell \in [K]$ across all times is given by 
    \begin{align*}
        \hat{g}_{\ell, \hat{\pi}} = 
        \begin{cases}
            \sum_{s \in \mathcal{S}} \hat{\rho}_{\ell, \hat{\pi}}(s) \sum_{\hat{a} \in \hat{\mathcal{A}}} \hat{\pi}_{\ell, \hat{\pi}}(\hat{a} \vert s) \hat{r}(\hat{a}), \\
            \quad\text{if } \ell \leq K-1,\\
            \sum_{s \in \mathcal{S}} \hat{\rho}_{\ell, \hat{\pi}}(s) \sum_{\hat{a} \in \hat{\mathcal{A}}} \hat{\pi}_{\ell, \hat{\pi}}(\hat{a} \vert s) \hat{r}(\hat{a}) +\\
            \quad \sum_{s, s' \in \mathcal{S}} \hat{\rho}_{\ell, \hat{\pi}}(s) \hat{P}_{\hat{\pi}}(s' \vert s) r_H(s'), \\
            \quad\text{if } \ell = K.\\
        \end{cases}
    \end{align*}
    Next, we are going to show that the long-run average reward $\hat{g}_{\hat{\pi}}$ achieved by $\hat{\pi}$ in each time step satisfies 
    \begin{align} \label{eq:g-hat-relations}
        \hat{g}_{\hat{\pi}} = \sum_{\ell = 1}^{K} \hat{g}_{\ell, \hat{\pi}}.
    \end{align}
    We note that
    \begin{align*}
        \hat{g}_{\hat{\pi}} =& \lim_{T \rightarrow \infty} \frac{1}{T} \\
        &\quad \mathbb{E}_{\hat{\pi}}\left[\sum_{t = 1}^T \left(\sum_{\ell = 1}^K \hat{r}(\hat{a}^t_{\ell}) \right) + r_H(s^t_{K+1}) \Bigg\lvert s^1_1 = s \right]\\
        =& \lim_{T \rightarrow \infty} \frac{1}{T} \\
        &\quad \left\{\sum_{\ell = 1}^{K} \sum_{t = 1}^T \sum_{s' \in \mathcal{S}} \left[(\hat{P}_{\hat{\pi}}^KP^{\text{sys}})^{t-1} (\hat{P}_{\hat{\pi}})^{\ell-1} \right](s' \vert s) \right. \cdot\\
        &\quad\quad \left. \sum_{\hat{a} \in \hat{\mathcal{A}}} \hat{\pi}_{\ell}(\hat{a} \vert s') \hat{r}(\hat{a}) \right.\\
        &\quad+ \left. \sum_{t = 1}^T \sum_{s' \in \mathcal{S}} \left[(\hat{P}_{\hat{\pi}}^KP^{\text{sys}})^{t-1} (\hat{P}_{\hat{\pi}})^{K} \right](s' \vert s) r_H(s') \right\}\\
        =& \sum_{\ell = 1}^K \sum_{s' \in \mathcal{S}} \left\{\lim_{T \rightarrow \infty} \frac{1}{T} \sum_{t = 1}^T \left[(\hat{P}_{\hat{\pi}}^K P^{\text{sys}})^{t-1} \right] (\hat{P}_{\hat{\pi}})^{\ell-1}\right\}(s' \vert s) \cdot\\
        &\quad\quad \sum_{\hat{a} \in \hat{\mathcal{A}}} \hat{\pi}_{\ell}(\hat{a} \vert s') \hat{r}(\hat{a})\\
        &\quad+ \sum_{s' \in \mathcal{S}} \left\{\lim_{T \rightarrow \infty} \frac{1}{T} \sum_{t = 1}^T \left[(\hat{P}_{\hat{\pi}}^K P^{\text{sys}})^{t-1} \right] (\hat{P}_{\hat{\pi}})^{K}\right\}(s' \vert s) \cdot\\
        &\quad\quad r_H(s')\\
        \stackrel{(a)}{=}& \sum_{\ell = 1}^K \sum_{s' \in \mathcal{S}} \left[ P^*_{\hat{\pi}} (\hat{P}_{\hat{\pi}})^{\ell-1}\right](s' \vert s) \sum_{\hat{a} \in \hat{\mathcal{A}}} \hat{\pi}_{\ell}(\hat{a} \vert s') \hat{r}(\hat{a}) \\
        &\quad+ \sum_{s' \in \mathcal{S}} \left[ P^*_{\hat{\pi}} (\hat{P}_{\hat{\pi}})^{K}\right](s' \vert s) r_H(s')\\
        =& \sum_{\ell = 1}^K \sum_{s' \in \mathcal{S}} \hat{\rho}_{\ell, \hat{\pi}}(s') \sum_{\hat{a} \in \hat{\mathcal{A}}} \hat{\pi}_{\ell}(\hat{a} \vert s') \hat{r}(\hat{a}) \\
        &\quad+ \sum_{s', s'' \in \mathcal{S}} \hat{\rho}_{K, \hat{\pi}}(s') \hat{P}_{\hat{\pi}}(s'' \vert s') r_H(s'')\\
        =& \sum_{\ell = 1}^{K} \hat{g}_{\ell, \hat{\pi}},
    \end{align*}
    where (a) holds due to the well-known result that the Cesaro limit equals to the ordinary limit whenever the latter one exists (see Theorem 43 on page 100 of \cite{hardy2024divergent}) and the fact that $(\hat{P}_{\hat{\pi}}^K P^{\text{sys}})^t$ converges to the stochastic matrix $P^*_{\hat{\pi}}$ as $t$ approaches infinity.

    Next, we argue that for each atomic step $\ell \in [K-1]$, the following series is well defined:
    \begin{align} \label{eq:h-def-standard-limit-single}
        \sum_{t = 1}^{\infty} \mathbb{E}_{\hat{\pi}}\left[\left( \hat{r}(\hat{a}_{\ell}^t) - \hat{g}_{\ell, \hat{\pi}} \right) \Bigg\lvert s^1_k = s \right],\quad \forall s \in \mathcal{S}.
    \end{align}
    We note that the series above can be equivalently written as: For all $s \in \mathcal{S}$,
    \begin{align*}
        & \sum_{t = 1}^{\infty} \mathbb{E}_{\hat{\pi}}\left[\left( \hat{r}(\hat{a}_{\ell}^t) - \hat{g}_{\ell, \hat{\pi}} \right) \Bigg\lvert s^1_k = s \right]\\
        =& \sum_{t = 1}^{\infty} \left\{ \sum_{s' \in \mathcal{S}} \left[(\hat{P}_{\hat{\pi}}^{K-k} P^{\text{sys}}) (\hat{P}_{\hat{\pi}}^K P^{\text{sys}})^{t-1} (\hat{P}_{\hat{\pi}})^{\ell-1} \right](s' \vert s) \cdot \right.\\
        &\quad \left.\sum_{\hat{a} \in \hat{\mathcal{A}}} \hat{\pi}_{\ell}(\hat{a} \vert s') \hat{r}(\hat{a}) - \hat{g}_{\ell, \hat{\pi}} \right\}\\
        =& \sum_{t = 1}^{\infty} \left\{ \sum_{s' \in \mathcal{S}} \left[(\hat{P}_{\hat{\pi}}^{K-k} P^{\text{sys}}) (\hat{P}_{\hat{\pi}}^K P^{\text{sys}})^{t-1} (\hat{P}_{\hat{\pi}})^{\ell-1} \right](s' \vert s) \right.\\
        &\quad \left. - \hat{\rho}_{\ell, \hat{\pi}}(s') \right\} \cdot \sum_{\hat{a} \in \hat{\mathcal{A}}} \hat{\pi}_{\ell}(\hat{a} \vert s') \hat{r}(\hat{a}).
    \end{align*}

    By Perron-Frobenius theorem (see Theorem 4.3.8 on page 160 of \cite{bremaud2013markov}), the term $(\hat{P}_{\hat{\pi}}^K P^{\text{sys}})^{t-1}$ converges to $P^*_{\hat{\pi}} = e \hat{\rho}^T_{1,\hat{\pi}}$ at a rate of $O(\gamma^t)$ for some constant $\gamma \in (0, 1)$. Hence, at the rate of $O(\gamma^t)$, the term \[
        (\hat{P}_{\hat{\pi}}^{K-k} P^{\text{sys}}) (\hat{P}_{\hat{\pi}}^K P^{\text{sys}})^{t-1} (\hat{P}_{\hat{\pi}})^{\ell-1}
    \] converges to 
    \begin{align*}
        &(\hat{P}_{\hat{\pi}}^{K-k} P^{\text{sys}}) P^*_{\hat{\pi}} (\hat{P}_{\hat{\pi}})^{\ell-1}\\ 
        =& (\hat{P}_{\hat{\pi}}^{K-k} P^{\text{sys}}) e \hat{\rho}^T_{1,\hat{\pi}} (\hat{P}_{\hat{\pi}})^{\ell-1}\\
        =& e\hat{\rho}^T_{1,\hat{\pi}} (\hat{P}_{\hat{\pi}})^{\ell-1} = e\hat{\rho}^T_{\ell, \hat{\pi}},
    \end{align*}
    which implies that the series
    \begin{align*}
        &\sum_{t = 1}^{\infty} \left\{ \sum_{s' \in \mathcal{S}} \left[(\hat{P}_{\hat{\pi}}^{K-k} P^{\text{sys}}) (\hat{P}_{\hat{\pi}}^K P^{\text{sys}})^{t-1} (\hat{P}_{\hat{\pi}})^{\ell-1} \right](s' \vert s) \right.\\
        &\quad \left.- \hat{\rho}_{\ell, \hat{\pi}}(s') \right\} \cdot \sum_{\hat{a} \in \hat{\mathcal{A}}} \hat{\pi}_{\ell}(\hat{a} \vert s') \hat{r}(\hat{a})
    \end{align*}
    converges and therefore \eqref{eq:h-def-standard-limit-single} is well-defined.

    Using a similar argument, we can obtain that the following series is also well defined:
    \begin{align} \label{eq:h-def-standard-limit-single-last}
        \sum_{t = 1}^{\infty} \mathbb{E}_{\hat{\pi}}\left[\left(\hat{r}(\hat{a}_K^t) + r_H(s_{K+1}^t) - \hat{g}_{K, \hat{\pi}} \right) \Bigg\lvert s^1_k = s \right],\quad \forall s \in \mathcal{S}.
    \end{align}
        
    Finally, we can reconstruct \eqref{eq:h-def-standard-limit-main} by using \eqref{eq:h-def-standard-limit-single} for all atomic steps:
    \begin{align*}
        & \mathbb{E}_{\hat{\pi}}\left[\sum_{\ell = k}^K \left( \hat{r}(\hat{a}_{\ell}^1) - \frac{1}{K} \hat{g}_{\hat{\pi}} \right) + r_H(s_{K+1}^1) \Bigg\lvert s^1_k = s \right] \\
        &\quad+ \sum_{\ell = 1}^{K-1} \sum_{t = 2}^{\infty} \mathbb{E}_{\hat{\pi}}\left[\left( \hat{r}(\hat{a}_{\ell}^t) - \hat{g}_{\ell, \hat{\pi}} \right) \Bigg\lvert s^1_k = s \right] \\
        &\quad+ \sum_{t = 2}^{\infty} \mathbb{E}_{\hat{\pi}}\left[\left( \hat{r}(\hat{a}_{K}^t) + r_H(s_{K+1}^t) - \hat{g}_{K,\hat{\pi}} \right) \Bigg\lvert s^1_k = s \right],\\
        &\quad \forall s \in \mathcal{S}\\
        =& \mathbb{E}_{\hat{\pi}}\left[\sum_{\ell = k }^K \left( \hat{r}(\hat{a}_{\ell}^1) - \frac{1}{K} \hat{g}_{\hat{\pi}} \right) + r_H(s_{K+1}^1) \Bigg\lvert s^1_k = s \right] \\
        &\quad+ \sum_{t = 2}^{\infty} \mathbb{E}_{\hat{\pi}}\left[\sum_{\ell = 1}^K \left( \hat{r}(\hat{a}_{\ell}^t) - \frac{1}{K} \hat{g}_{\hat{\pi}} \right) + r_H(s_{K+1}^t) \Bigg\lvert s^1_k = s \right],\\
        &\quad \forall s \in \mathcal{S} \tag{a}\\
        =& \hat{h}_{k, \hat{\pi}}(s), \quad \forall s \in \mathcal{S},
    \end{align*}
    where (a) is obtained by interchanging the sum $\sum_{\ell = 1}^K$ and the infinite sum $\sum_{t = 2}^{\infty}$, as we have already shown that the series \eqref{eq:h-def-standard-limit-single}-\eqref{eq:h-def-standard-limit-single-last} are well-defined. Hence, we can conclude that $h_{k, \hat{\pi}}(s)$ given in \eqref{eq:h-def-standard-limit-main} is well defined for all $s \in \mathcal{S}$.
    \hfill $\square$
\end{proof}

\begin{proof}{Proof of Proposition \ref{prop:h-def-equivalence}}
    For any $s \in \mathcal{S}$, we have
    \begin{align*}
        &\hat{h}_{k, \hat{\pi}}(s) =\\
        &\mathbb{E}_{\hat{\pi}}\left[ \sum_{\ell = k }^K \left( \hat{r}(\hat{a}_{\ell}^1) - \frac{1}{K} \hat{g}_{\hat{\pi}} \right) + r_H(s_{K+1}^1) \Bigg\lvert s^1_k = s \right] \\
        &+ \sum_{t = 2}^{\infty} \sum_{\ell = 1}^K \mathbb{E}_{\hat{\pi}}\left[\left( \hat{r}(\hat{a}_{\ell}^t) - \frac{1}{K} \hat{g}_{\hat{\pi}} \right) + r_H(s_{K+1}^t) \Bigg\lvert s^1_k = s \right]\\
        =& \mathbb{E}_{\hat{\pi}}\left[ \sum_{\ell = k }^K \left( \hat{r}(\hat{a}_{\ell}^1) - \frac{1}{K} \hat{g}_{\hat{\pi}} \right) + r_H(s_{K+1}^1) \Bigg\lvert s^1_k = s \right]\\ 
        &+ \sum_{\ell = 1}^{K-1} \sum_{t = 2}^{\infty} \mathbb{E}_{\hat{\pi}}\left[\left( \hat{r}(\hat{a}_{\ell}^t) - \hat{g}_{\ell, \hat{\pi}} \right) \Bigg\lvert s^1_k = s \right] \\
        &+ \sum_{t = 2}^{\infty} \mathbb{E}_{\hat{\pi}}\left[\left(\hat{r}(\hat{a}_K^t) + r_H(s_{K+1}^t) - \hat{g}_{K, \hat{\pi}} \right) \Bigg\lvert s^1_k = s \right].
    \end{align*}
    Due to the well-known result that the Cesaro limit equals to the ordinary limit whenever the latter one exists (see Theorem 43 on page 100 of \cite{hardy2024divergent}) and the fact that we have already shown that \eqref{eq:h-def-standard-limit-main} is well-defined, the above equation is equivalent to
    \begin{align*}
        &\mathbb{E}_{\hat{\pi}}\left[ \sum_{\ell = k }^K \left( \hat{r}(\hat{a}_{\ell}^1) - \frac{1}{K} \hat{g}_{\hat{\pi}} \right) + r_H(s_{K+1}^1) \Bigg\lvert s^1_k = s \right] \\
        &+ \sum_{q = 1}^{K-1} \lim_{D \rightarrow \infty} \frac{1}{D} \sum_{T = 2}^D  \mathbb{E}_{\hat{\pi}}\left[\sum_{t = 2}^{T} \left( \hat{r}(\hat{a}_{q}^t) - \hat{g}_{q, \hat{\pi}} \right) \Bigg\lvert s^1_k = s \right] \\
        &+ \lim_{D \rightarrow \infty} \frac{1}{D} \sum_{T = 2}^D  \mathbb{E}_{\hat{\pi}}\left[\sum_{t = 2}^{T} \left(\hat{r}(\hat{a}_K^t) + r_H(s_{K+1}^t) - \hat{g}_{K, \hat{\pi}} \right) \Bigg\lvert s^1_k = s \right],
    \end{align*}
    which equals to
    \begin{align*}
        &\mathbb{E}_{\hat{\pi}}\left[ \sum_{\ell = k }^K \left( \hat{r}(\hat{a}_{\ell}^1) - \frac{1}{K} \hat{g}_{\hat{\pi}} \right) + r_H(s_{K+1}^1) \Bigg\lvert s^1_k = s \right] \\
        &+ \lim_{D \rightarrow \infty} \frac{1}{DK} \sum_{T = 2}^D\\
        &\quad\left\{\sum_{q = 1}^{K-1}   \mathbb{E}_{\hat{\pi}}\left[\sum_{t = 2}^{T} K \left( \hat{r}(\hat{a}_{q}^t) - \hat{g}_{q, \hat{\pi}} \right) \Bigg\lvert s^1_k = s \right] \right.\\ 
        &+ \left. \mathbb{E}_{\hat{\pi}}\left[\sum_{t = 2}^{T} K \left(\hat{r}(\hat{a}_K^t) + r_H(s_{K+1}^t) - \hat{g}_{K, \hat{\pi}} \right) \Bigg\lvert s^1_k = s \right] \right\}.
    \end{align*}
    Using the relationship in \eqref{eq:g-hat-relations}, we obtain
    \begin{align*}
        &\mathbb{E}_{\hat{\pi}}\left[ \sum_{\ell = k }^K \left( \hat{r}(\hat{a}_{\ell}^1) - \frac{1}{K} \hat{g}_{\hat{\pi}} \right) + r_H(s_{K+1}^1) \Bigg\lvert s^1_k = s \right] \\
        &+ \lim_{D \rightarrow \infty} \frac{1}{DK} \sum_{T = 2}^D \\
        &\quad \left\{\sum_{q = 1}^{K-1}   \mathbb{E}_{\hat{\pi}}\left[\sum_{t = 2}^{T} K \left( \hat{r}(\hat{a}_{q}^t) - \frac{1}{K}\hat{g}_{\hat{\pi}} \right) \Bigg\lvert s^1_k = s \right] \right.\\ 
        &+ \left. \mathbb{E}_{\hat{\pi}}\left[\sum_{t = 2}^{T} K \left(\hat{r}(\hat{a}_K^t) + r_H(s_{K+1}^t) - \frac{1}{K}\hat{g}_{\hat{\pi}} \right) \Bigg\lvert s^1_k = s \right] \right\}.
    \end{align*}
    This can be equivalently written as:
    \begin{align*}
        & \mathbb{E}_{\hat{\pi}}\left[ \sum_{\ell = k }^K \left( \hat{r}(\hat{a}_{\ell}^1) - \frac{1}{K} \hat{g}_{\hat{\pi}} \right) + r_H(s_{K+1}^1) \Bigg\lvert s^1_k = s \right] \\
        &+ \lim_{D \rightarrow \infty} \frac{1}{DK} \sum_{T = 2}^D\\
        &\quad \left\{ \sum_{q = 1}^{K} \mathbb{E}_{\hat{\pi}}\left[\sum_{t = 2}^{T-1} \left(\sum_{\ell = 1}^{K-1} \left( \hat{r}(\hat{a}_{\ell}^t) - \frac{1}{K}\hat{g}_{\hat{\pi}} \right) + r_H(s_{K+1}^t) \right) \right.\right. \\
        &\quad \left. + \sum_{\ell = 1}^{q} \left( \hat{r}(\hat{a}_{\ell}^T) - \frac{1}{K}\hat{g}_{\hat{\pi}} \right) \right.\\
        &\quad+ \left. \mathds{1}\{q=K\}\left(\hat{r}(\hat{a}_K^T) + r_H(s_{K+1}^T) - \frac{1}{K}\hat{g}_{\hat{\pi}} \right) \Bigg\lvert s^1_k = s \right]\\
        &\quad \left. + \sum_{q = 1}^{K-1}   \mathbb{E}_{\hat{\pi}}\left[(q-1) \left( \hat{r}(\hat{a}_{q}^T) - \frac{1}{K}\hat{g}_{\hat{\pi}} \right) \Bigg\lvert s^1_k = s \right] \right.\\
        &\quad+ \left. \mathbb{E}_{\hat{\pi}}\left[K\left(\hat{r}(\hat{a}_K^T) + r_H(s_{K+1}^T) - \frac{1}{K}\hat{g}_{\hat{\pi}} \right) \Bigg\lvert s^1_k = s \right] \right\}.
    \end{align*}
    Using the fact that 
    \begin{align*}
        \hat{g}_{\ell,\hat{\pi}} = 
        \begin{cases}
            \lim_{D \rightarrow \infty} \frac{1}{D} \sum_{T = 2}^D \mathbb{E}_{\hat{\pi}}\left[\hat{r}(\hat{a}^T_{\ell}) \vert s^1_k = s \right],\\
            \quad \text{if } \ell \leq K-1,\\
            \lim_{D \rightarrow \infty} \frac{1}{D} \sum_{T = 2}^D \mathbb{E}_{\hat{\pi}}\left[\hat{r}(\hat{a}^T_{K}) + r_H(s_{K+1}^T) \vert s^1_k = s \right],\\
            \quad \text{if } \ell = K,\\
        \end{cases}
    \end{align*}
    we obtain that the equation above is equivalent to:
    \begin{align*}
        & \mathbb{E}_{\hat{\pi}}\left[ \sum_{\ell = k }^K \left( \hat{r}(\hat{a}_{\ell}^1) - \frac{1}{K} \hat{g}_{\hat{\pi}} \right) + r_H(s_{K+1}^1) \Bigg\lvert s^1_k = s \right] \\
        &+ \lim_{D \rightarrow \infty} \frac{1}{DK} \sum_{T = 2}^D\\
        &\quad \left\{ \sum_{q = 1}^{K} \mathbb{E}_{\hat{\pi}}\left[\sum_{t = 2}^{T-1} \left(\sum_{\ell = 1}^K \left( \hat{r}(\hat{a}_{\ell}^t) - \frac{1}{K}\hat{g}_{\hat{\pi}} \right) + r_H(s_{K+1}^t) \right) \right.\right.\\
        &\quad \left. + \sum_{\ell = 1}^{q} \left( \hat{r}(\hat{a}_{\ell}^T) - \frac{1}{K}\hat{g}_{\hat{\pi}} \right) \right.\\
        &\quad+ \left. \left. \mathds{1}\{q=K\}\left(\hat{r}(\hat{a}_K^T) + r_H(s_{K+1}^T) - \frac{1}{K}\hat{g}_{\hat{\pi}} \right) \Bigg\lvert s^1_k = s \right]\right\}\\
        & + \frac{1}{K} \sum_{\ell = 1}^{K} (\ell-1)\left(\hat{g}_{\ell, \hat{\pi}} - \frac{1}{K} \hat{g}_{\hat{\pi}} \right).
    \end{align*}
    Finally, we can conclude the proof by setting the scalar $\kappa_{\hat{\pi}}$ to be: \[
        \kappa_{\hat{\pi}} := -\frac{1}{K} \sum_{\ell = 1}^{K} (\ell-1)\left(\hat{g}_{\ell, \hat{\pi}} - \frac{1}{K} \hat{g}_{\hat{\pi}} \right).
    \]
    \hfill $\square$
\end{proof}

\clearpage




%
%
%


\bibliographystyle{unsrtnat} 
\bibliography{references} 

\end{document}